\documentclass[10pt,A4paper]{article}
\usepackage{amssymb}
\usepackage{amsmath}
\usepackage{amsthm}
\usepackage{latexsym}
\usepackage[dvips]{epsfig}
\usepackage{mathrsfs}
\usepackage{eufrak}
\usepackage{bm}
\usepackage{tikz}
\usepackage{authblk}

\theoremstyle{plain}
\newtheorem{proposition}{Proposition}
\newtheorem{lemma}{Lemma}
\newtheorem{theorem}{Theorem}
\newtheorem{assumption}{Assumption}

\newtheorem{definition}{Definition}

\newtheorem{remark}{Remark}

\def\bmg{{\bm g}}
\def\bmh{{\bm h}}

\def\bmzero{{\bm 0}}
\def\bmone{{\bm 1}}

\def\bmA{{\bm A}}
\def\bmB{{\bm B}}
\def\bmC{{\bm C}}
\def\bmD{{\bm D}}
\def\bmE{{\bm E}}
\def\bmF{{\bm F}}

\def\bmK{{\bm K}}

\def\bmQ{{\bm Q}}

\def\bmeta{{\bm \eta}}
\def\bmzeta{{\bm\zeta}}
\def\bmxi{{\bm \xi}}
\def\bmchi{{\bm \chi}}

\def\bmkappa{{\bm \kappa}}

\def\bmmu{{\bm \mu}}
\def\bmnu{{\bm \nu}}

\def\bmpartial{{\bm \partial}}

\def\ho{{\widehat{o}}}

\newcounter{mnotecount}

\newcommand{\mnotex}[1]
{\protect{\stepcounter{mnotecount}}$^{\mbox{\footnotesize $\bullet$\themnotecount}}$ 
\marginpar{
\raggedright\tiny\em
$\!\!\!\!\!\!\,\bullet$\themnotecount: #1} }

\begin{document}


\title{\textbf{A geometric invariant characterising initial data for
    the Kerr-Newman spacetime}}

\author[,1]{Michael J. Cole \footnote{E-mail address:{\tt m.j.cole@qmul.ac.uk}}}
\author[,1]{Juan A. Valiente Kroon \footnote{E-mail address:{\tt j.a.valiente-kroon@qmul.ac.uk}}}
\affil[1]{School of Mathematical Sciences, Queen Mary, University of London,
Mile End Road, London E1 4NS, United Kingdom.}

\maketitle

\begin{abstract}
We describe the construction of a geometric invariant characterising
initial data for the Kerr-Newman spacetime. This geometric invariant
vanishes if and only if the initial data set corresponds to exact
Kerr-Newman initial data, and so characterises this type of data. We
first illustrate the characterisation of the Kerr-Newman spacetime in
terms of Killing spinors. The space spinor formalism is then used to
obtain a set of four independent conditions on an initial Cauchy
hypersurface that guarantee the existence of a Killing spinor on the
development of the initial data. Following a similar analysis in the
vacuum case, we study the properties of solutions to the approximate
Killing spinor equation and use them to construct the geometric
invariant.
\end{abstract}

\section{Introduction}

The Kerr-Newman solution to the Einstein-Maxwell equations, describing
a stationary charged rotating black hole, is one of the most interesting and well
studied exact solutions in General Relativity, and yet there still
remain several unresolved questions. For example, the current family
of uniqueness results regarding the Kerr-Newman solution contain
assumptions on the spacetime that are often considered too
restrictive, such as analyticity ---see e.g. \cite{ChrCos12} for a
review on the subject. Also, although there has been
significant progress on the linear stability of the Kerr-Newman
solution, the question of non-linear stability has been far more
stubborn ---see e.g. \cite{DafRod10} for a discussion on this topic.

Making progress on these unresolved questions concerning electrovacuum
blackholes provides the motivation for finding characterisations of
the Kerr-Newman spacetime. Different methods for characterising the
exact solution can be tailored to emphasise specific properties, and so address each of these unresolved properties
directly. One such characterisation is expressed in terms of Killing spinors,
closely related to Killing-Yano tensors, which represent hidden
symmetries of the spacetime. These symmetries cannot be expressed in
terms of isometries of the spacetime. It has been shown in \cite{ColVal16a}
that an asymptotically flat electrovacuum spacetime admitting a
Killing spinor which satisfies a certain \emph{alignment condition}
with the Maxwell field must be isometric to the Kerr-Newman spacetime
---see Theorem \ref{Theorem:KerrNewmanCharacterisation} below.

Once the motivation for a characterisation of the Kerr-Newman
spacetime in terms of Killing spinors has been established, it is
useful to investigate how the existence of such a spinor can be
expressed in terms of initial data.  The initial value problem in
General Relativity has played a crucial role in the systematic
analysis of the properties of generic solutions to the Einstein field
equations ---see e.g. \cite{Fou52,Ren08,Rin09b}. It also provides the framework necessary for numerical simulations of
spacetimes to be performed ---see e.g. \cite{Alc08,BauSha10}. 

Representing symmetries of a spacetime in terms of conditions on an
initial hypersurface is not a new idea; the \emph{Killing initial data
  (KID) equations} ---see e.g. \cite{BeiChr97b}--- are conditions on a spacelike Cauchy
surface $\mathcal{S}$ which guarantee the existence of a Killing vector in
the resulting development of the initial data. Thus, isometries of the whole
spacetime can be encoded at the level of initial data. The resulting
conditions form a system of overdetermined equations, so do not
necessarily admit a solution for an arbitrary initial data set. In
fact, it has been shown that the KID equations are non-generic, in the
sense that generic solutions of the vacuum constraint Einstein
equations do not possess any global or local spacetime Killing
vectors ---see \cite{BeiChrSch05}. An analogous construction can, in
principle, be performed for Killing spinors. This analysis has been performed for the vacuum case giving explicitly the conditions relating the Killing spinor candidate
and the Weyl curvature of the spacetime ---see \cite{GarVal08a} and also \cite{BaeVal10b}. These conditions are, like
the KID equations, an overdetermined system and so do not necessarily
admit a solution for an arbitrary initial surface. However, in
\cite{BaeVal10a,BaeVal10b} it has been shown that given an
asymptotically Euclidean hypersurface it is always possible to
construct a
\emph{Killing spinor candidate} which, whenever there exists a Killing
spinor in the development, coincides with the restriction
of the Killing spinor to the initial hypersurface. This approximate
Killing spinor is obtained by solving a linear second order elliptic equation
which is the Euler-Lagrange equation of a certain functional over
$\mathcal{S}$.  The approximate Killing spinor can be used to
construct a geometric invariant which in some way parametrises the
deviation of the initial data set from Kerr initial data. Variants of
the basic construction in \cite{BaeVal10b} have been given in
\cite{BaeVal11b,BaeVal12}. 

The purpose of this article is to extend the analysis of
\cite{BaeVal10b} to the
electrovacuum case. In doing so, we rely on the characterisation of
the Kerr-Newman spacetime given in \cite{ColVal16a} which, in turn,
builds upon the characterisation provided in \cite{Mar00} for the
vacuum case and \cite{Won09} for the electrovacuum case. As a result
of our analysis we find that the Killing spinor initial data
equations remain largely unchanged, with extra conditions ensuring
that the electromagnetic content of the spacetime inherits the
symmetry of the Killing spinor. These electrovacuum Killing spinor
equations, together with an appropriate approximate Killing spinor,
are used to
construct an invariant expressed in terms of suitable integrals over
the hypersurface $\mathcal{S}$ whose vanishing characterises in a
necessary and sufficient manner initial
data for the Kerr-Newman spacetime. Our main result, in this respect,
is given in Theorem \ref{MainTheorem:CharacterisationKerrNewmanData}.

\subsection*{Overview of the article}
Section \ref{Section:KillingSpinorsElectrovacuum} provides a brief
overview of the theory of Killing spinors in electrovacuum
spacetimes. Section \ref{Section:EvolutionKillingSpinors} discusses
the evolution equations governing the propagation of the Killing
spinor equation in an electrovacuum spacetime. The main conclusion
from this analysis is that the resulting system is homogeneous in a
certain set of zero-quantities. The trivial data for these equations
gives rise to the conditions implying the existence of a Killing
spinor in the development of some initial hypersurface. In Section
\ref{Section:KillingSpinorDataEquation} a space-spinor formalism is
used to reexpress these conditions in terms of quantities intrinsic to
the initial hypersurface. In addition, in this section the
interdependence between the various conditions is analysed and a
\emph{minimal} set of \emph{Killing spinor data equations} is
obtained. Section \ref{Section:ApproximateKillingEquation}
introduces the notion of approximate Killing spinors for electrovacuum
initial data sets and discusses som basic ellipticity properties of
the associated approximate Killing spinor equation. Section
\ref{Section:AsymptoticallyEuclideanManifolds} discusses the
solvability of the approximate Killing spinor equation in a class of
asymptotically Euclidean manifolds. Finally, Section
\ref{Section:Invariant} brings together the analyses in the various
section to construct a geometric invariant characterising initial data
for the Kerr-Newman spacetime. The main result of this article is
given in Theorem \ref{MainTheorem:CharacterisationKerrNewmanData}.  

\subsection*{Notation and conventions}
Let $(\mathcal{M},\bmg,\bmF)$ denote an electrovacuum spacetime ---i.e. a
solution to the Einstein-Maxwell field equations. The signature of the metric in this article
will be $(+,-,-,-)$, to be consistent with most of the existing
literature using spinors. We use the spinorial conventions of
\cite{PenRin84}. The lowercase Latin letters $a,\, b,\, c, \ldots$ are
used as abstract spacetime tensor indices while the uppercase letters
$A,\,B,\,C,\ldots$ will serve as abstract spinor indices. The Greek
letters $\mu, \, \nu, \, \lambda,\ldots$ will be used as spacetime
coordinate indices while $\alpha,\,\beta,\,\gamma,\ldots$ will serve
as spatial coordinate indices. Finally $\bmA,\,\bmB,\,\bmC,\ldots$
will be used as spinorial frame indices.

The conventions for
the spinorial curvature are set via the expressions
\begin{equation}
\square_{AB}\mu_C = \Psi_{ABCD}\mu^D - 2 \Lambda
\mu_{(A}\epsilon_{B)C}, \qquad \square_{A'B'}\mu_C =
\Phi_{ACA'B'}\mu^A. 
\label{BoxCommutators}
\end{equation}
We systematically use of the following expression for the (once
contracted) second derivative of a spinor:
\begin{equation}
\nabla_{AQ'}\nabla_B{}^{Q'} = \frac{1}{2}\epsilon_{AB}\square + \square_{AB}.
\label{SecondDerivativeDecomposition}
\end{equation}

\section{Killing spinors in electrovacuum spacetimes}
\label{Section:KillingSpinorsElectrovacuum}
In this section we provide a systematic exposition of the properties
of Killing spinors in an electrovacuum spacetime.

\subsection{The Einstein-Maxwell equations}
Using standard spinorial notation, the Einstein-Maxwell equations are
given by
\begin{subequations}
\begin{eqnarray}
&\Phi_{ABA'B'} = 2 \phi_{AB} \bar{\phi}_{A'B'}, \qquad \Lambda =0, & \label{EinsteinMaxwell1}\\
& \nabla^A{}_{A'} \phi_{AB}=0. \label{EinsteinMaxwell2}
\end{eqnarray}
\end{subequations}
In particular, from the Maxwell equation \eqref{EinsteinMaxwell2} it
follows that
\[
  \nabla_{A'B} \phi_{CD}=\nabla_{A'(B}\phi_{CD)}.
\]

The Bianchi identity is given by
\[
\nabla^A{}_{B'} \Psi_{ABCD} = \nabla_{(B}{}^{A'} \Phi_{CD)A'B'},
\]
or, more explicitly
\begin{equation}
\nabla^{A}{}_{A'}\Psi_{ABCD}=2\bar{\phi}_{A'B'}\nabla_{B}{}^{B'}\phi_{CD}.
\label{BianchiIdentity}
\end{equation}

Given an electrovacuum spacetime, applying the derivative
$\nabla^{A'}{}_{C}$ to the Maxwell equation in the form $\nabla^A{}_{A'}
\phi_{AB}=0$ one obtains, after some standard manipulations, the
following wave equation for the Maxwell spinor:
\begin{equation}
\square \phi_{AB}= 2 \Psi_{ABCD}\phi^{CD}.
\label{WaveEquationMaxwellSpinor}
\end{equation}

\subsection{Killing spinors}
A \emph{Killing spinor} $\kappa_{AB}=\kappa_{(AB)}$ in an electrovacuum
spacetime $(\mathcal{M},\bmg,\bmF)$  is a solution to the Killing spinor equation 
\begin{equation}
\nabla_{A'(A} \kappa_{BC)}=0.
\label{KillingSpinorEquation}
\end{equation}

In the sequel, a prominent role will played by the integrability
conditions implied by the Killing spinor equation. More precisely, one
has the following:

\begin{lemma}
\label{Lemma:IntegrabilityConditionsKillingSpinorEquations}
Let $(\mathcal{M},\bmg,\bmF)$ denote an electrovacuum spacetime
endowed with a Killing spinor $\kappa_{AB}$. Then $\kappa_{AB}$
satisfies the integrability conditions:
\begin{subequations}
\begin{eqnarray}
&& \kappa_{(A}{}^Q \Psi_{BCD)}{}_Q=0, \label{IntegrabilityCondition} \\
&& \square \kappa_{AB} +\Psi_{ABCD}\kappa^{CD}=0. \label{WaveEquationKappa}
\end{eqnarray}
\end{subequations}
\end{lemma}

\begin{proof}
The integrability conditions follow from applying the derivative
$\nabla_D{}^{A'}$ to the Killing spinor equation
\eqref{KillingSpinorEquation}, then using the identity
\eqref{SecondDerivativeDecomposition} together with the box
commutators \eqref{BoxCommutators} and finally decomposing the
resulting expression into its irreducible terms ---the only
non-trivial trace yields equation \eqref{WaveEquationKappa} while
the completely symmetric part gives equation
\eqref{IntegrabilityCondition}. 
\end{proof}

\begin{remark}
{\em Observe that although every solution to the Killing
spinor equation \eqref{KillingSpinorEquation} satisfies the wave
equation \eqref{WaveEquationKappa}, the converse is not true. In what
follows, a  symmetric spinor satisfying equation
\eqref{WaveEquationKappa}, but not necessarily equation
\eqref{KillingSpinorEquation}, will be called a \emph{Killing spinor
  candidate}. This notion will play a central role in our subsequent
analysis ---in particular, we will be concerned with the question of
the further conditions that need to be imposed on a Killing spinor
candidate to be an actual Killing spinor. }
\end{remark}

\medskip
A well-known property of Killing spinors in a vacuum spacetime is that
the spinor
\begin{equation}
\xi_{AA'} \equiv \nabla^Q{}_{A'} \kappa_{QA}
\label{KillingVectorCandidate}
\end{equation}
is the counterpart of a (possibly complex) Killing vector $\xi^a$. A
similar property holds for electrovacuum spacetimes ---however,
a further condition is required on the Killing spinor.

\begin{lemma}
Let $(\mathcal{M},\bmg,\bmF)$ denote an electrovacuum spacetime
endowed with a Killing spinor $\kappa_{AB}$. Then $\xi_{AA'}$ as
defined by equation \eqref{KillingVectorCandidate} is the spinorial
counterpart of a Killing vector $\xi^a$ if and only if
\begin{equation}
\kappa_{(A}{}^Q\phi_{B)Q}=0.
\label{MatterAlignmentCondition}
\end{equation}
\end{lemma}

\begin{proof}
The proof follows by direct substitution of the definition
\eqref{KillingVectorCandidate} into the derivative $\nabla_{AA'}
\xi_{BB'}$. Again, using the box commutators
\eqref{BoxCommutators} one obtains, after some manipulations that
\[
\nabla_{AA'}
\xi_{BB'} + \nabla_{BB'}\xi_{AA'} = 12 \bar{\phi}_{A'B'}\,
\kappa_{(A}{}^Q\phi_{B)Q},
\]
from which the result follows.
\end{proof}

\begin{remark}
 {\em Condition \eqref{MatterAlignmentCondition} implies
that the Killing spinor $\kappa_{AB}$ and the Maxwell spinor
$\phi_{AB}$ are proportional to each other ---thus, in what follows we
refer to \eqref{MatterAlignmentCondition} as the \emph{matter
  alignment condition}.}
\end{remark}

\begin{remark}
{\em In the sequel we will refer to a spinor $\xi_{AA'}$
obtained from a symmetric spinor $\kappa_{AB}$ using expression
\eqref{KillingVectorCandidate} (not necessarily a Killing spinor) as the
\emph{Killing vector candidate} associated to $\kappa_{AB}$.}
\end{remark}

\subsection{Zero-quantities}
In order to investigate the the consequences of the Killing spinor
equation \eqref{KillingSpinorEquation} in a more systematic manner it
is convenient to introduce the following \emph{zero-quantities}:
\begin{subequations}
\begin{eqnarray}
&& H_{A'ABC} \equiv 3 \nabla_{A'(A} \kappa_{BC)}, \label{DefinitionH}\\ 
&& S_{AA'BB'} \equiv \nabla_{AA'} \xi_{BB'} + \nabla_{BB'} \xi_{AA'}, \label{DefinitionS}\\
&& \Theta_{AB} \equiv 2 \kappa_{(A}{}^Q \phi_{B)Q}. \label{DefinitionTheta}
\end{eqnarray}
\end{subequations}
Observe that if $H_{A'ABC}=0$ then $\kappa_{AB}$ is a Killing spinor. Similarly, if
$S_{AA'BB'}=0$ then $\xi_{AA'}$ is the spinor counterpart of a
Killing vector, while if  $\Theta_{AB}=0$ then the matter alignment condition
\eqref{MatterAlignmentCondition} holds. 

\medskip
The decomposition in irreducible components of $\nabla_{AA'}
\kappa_{BC}$ can be expressed in terms of $H_{A'ABC}$ and $\xi_{AA'}$
as
\begin{equation}
\label{Decomposition:DerivativeKappa}
\nabla_{AA'} \kappa_{BC} = \frac{1}{3}H_{A'ABC}
-\frac{2}{3}\epsilon_{A(B} \xi_{C)A'}. 
\end{equation}
Similarly, a further computation shows that for $\xi_{AA'}$ as given
by equation \eqref{KillingVectorCandidate} one has the decomposition
\begin{equation}
\nabla_{AA'}\xi_{BB'} = \bar{\eta}_{A'B'} \epsilon_{AB} +
\eta_{AB}\epsilon_{A'B'} + \frac{1}{2}S_{(AB)(A'B')}
\label{DecompositionDerivativeKillingVectorCandidate}
\end{equation}
where
\[
\eta_{AB} \equiv \frac{1}{2} \nabla_{AQ'} \xi_B{}^{Q'}.
\]
If $\xi_{AA'}$ is a real Killing vector then the spinor $\eta_{AB}$
encodes the information of the so-called \emph{Killing form}.

\begin{remark}
{\em From equation
\eqref{DecompositionDerivativeKillingVectorCandidate} it readily
follows by contraction that
\[
\nabla^{AA'}\xi_{AA'} =0
\]
independently of whether the alignment condition
\eqref{MatterAlignmentCondition} holds or not ---i.e. the Killing
vector candidate $\xi_{AA'}$ defined by equation
\eqref{KillingVectorCandidate} is always divergence free. This
observation, in turn, implies that
\[
S_{AA'}{}^{AA'}=0,
\]
so that one has the symmetry
\begin{equation}
\label{Symmetries:S}
S_{AA'BB'}=S_{(AB)(A'B')}.
\end{equation}}
\end{remark}

\begin{remark}
{\em The zero-quantities introduced in
equations \eqref{DefinitionH}-\eqref{DefinitionTheta} are a helpful
bookkeeping device. In particular, a calculation analogous to that of
the proof of Lemma
\ref{Lemma:IntegrabilityConditionsKillingSpinorEquations} shows that
\begin{eqnarray*}
&& \nabla_{(A}{}^{A'} H_{|A'|BCD)} =  -6 \Psi_{Q(ABC} \kappa_{D)}{}^Q, \\ 
&& \nabla^{AA'} H_{A'ABC} = 2 \big( \square \kappa_{BC} + \Psi_{BCPQ}\kappa^{PQ} \big).
\end{eqnarray*}}
\end{remark}

Thus, the integrability conditions of Lemma
\ref{Lemma:IntegrabilityConditionsKillingSpinorEquations} can be written,
alternatively, as
\[
\nabla_{(A}{}^{A'} H_{|A'|BCD)} =0, \qquad \nabla^{AA'} H_{A'ABC} =0.
\]
In particular, observe that if $\kappa_{AB}$ is a Killing spinor
candidate, then the zero quantity $H_{A'ABC}$ is divergence free. 

\subsection{A characterisation of Kerr-Newman in terms of spinors}
\label{Section:CharacterisationKerrNewman}

The following definition will play an important role in our subsquent
analysis:

\begin{definition}
\label{Definition:StationaryEnd}
A stationary asymptotically flat 4-end in an electrovacuum spacetime $(\mathcal{M},\bmg,\bmF)$ is an open
submanifold $\mathcal{M}_\infty\subset \mathcal{M}$ diffeomorphic to
$I\times ( \mathbb{R}^3\setminus \mathcal{B}_R)$ where $I\subset
\mathbb{R}$ is an open interval and $\mathcal{B}_R$ is a closed ball
of radius $R$. In the local coordinates $(t,x^\alpha)$ defined by the
diffeomorphism the components $g_{\mu\nu}$ and $F_{\mu\nu}$ of the
metric $\bmg$ and the Faraday tensor $\bmF$ 
satisfy 
\begin{subequations}
\begin{eqnarray}
&& |g_{\mu\nu} -\eta_{\mu\nu}| +|r \partial_\alpha g_{\mu\nu} | \leq C
r^{-1}, \label{StationaryEndCondition1}\\
&& |F_{\mu\nu}| + |r \partial_\alpha F_{\mu\nu}| \leq C' r^{-2}, \label{StationaryEndCondition2}\\
&& \partial_t g_{\mu\nu}=0, \label{StationaryEndCondition3}\\
&& \partial_t F_{\mu\nu} =0, \label{StationaryEndCondition4}
\end{eqnarray}
\end{subequations}
where $C$ and $C'$ are positive constants, $r\equiv (x^1)^2 +
(x^2)^2 + (x^3)^2,$ and $\eta_{\mu\nu}$ denote the components of the
Minkowski metric in Cartesian coordinates. 
\end{definition}

\begin{remark}
{\em  It follows from condition
\eqref{StationaryEndCondition3} in Definition
\ref{Definition:StationaryEnd} that the stationary asymptotically flat
end $\mathcal{M}_\infty$ is endowed with a Killing vector $\xi^a$
which takes the form $\bmpartial_t$ ---a so-called \emph{time
translation}. From condition \eqref{StationaryEndCondition4} one has
that the electromagnetic fields inherits the symmetry of the spacetime ---that is
$\mathcal{L}_\xi \bmF=0$, with $\mathcal{L}_\xi$ the Lie derivative
along $\xi^a$.}
\end{remark}

Of particular interest will be those stationary asymptotically flat ends \emph{generated by a Killing spinor}:

\begin{definition}
\label{Definition:KillingGeneratedEnd}
A stationary asymptotically flat end $\mathcal{M}_\infty\subset \mathcal{M}$ in an
electrovacuum spacetime $(\mathcal{M},\bmg,\bmF)$ endowed with a
Killing spinor $\kappa_{AB}$ is said to be generated by a
Killing spinor if the spinor $\xi_{AA'}\equiv
\nabla^B{}_{A'} \kappa_{AB}$ is the spinorial
counterpart of the Killing vector $\xi^a$. 
\end{definition}

\begin{remark}
{\em Stationary spacetimes have a natural definition of mass in terms of
the Killing vector $\xi^{a}$ that generates the isometry  ---the
so-called \emph{Komar mass} $m$ defined as
\[
m\equiv
-\frac{1}{8\pi}\lim_{r\rightarrow\infty}\int_{S_{r}}\epsilon_{abcd}\nabla^c\xi^d \mbox{d}S^{ab}
\]
where $S_{r}$ is the sphere of radius r centred at $r=0$ and
$\mbox{d}S^{ab}$ is the binormal vector to $S_r$.  Similarly, one can
define the \emph{total electromagnetic charge} of the spacetime via the integral
\[
q= -\frac{1}{4\pi}\lim_{r\rightarrow\infty}\int_{S_{r}} F_{ab} \mbox{d}S^{ab}.
\]}
\end{remark}

\begin{remark}
{\em In the asymptotic region the components of the metric can be written
in the form
\begin{eqnarray*}
&& g_{00} = 1- \frac{2 m}{r} + O(r^{-2}), \\
&& g_{0\alpha} = \frac{4 \epsilon_{\alpha\beta\gamma} S^\beta
  x^\gamma}{r^3} +O(r^{-3}), \\
&& g_{\alpha\beta} = -\delta_{\alpha\beta} + O(r^{-1}),
\end{eqnarray*}
where $m$ is the Komar mass of $\xi^a$ in the end
$\mathcal{M}_\infty$,  $\epsilon_{\alpha\beta\gamma}$ is the flat
rank 3 totally antisymmetric tensor and $S^\beta$ denotes a
3-dimensional tensor with constant entries. For the components of the Faraday
tensor one has that
\begin{eqnarray*}
&& F_{0\alpha} = \frac{q}{r^2} + O(r^{-3}), \\
&& F_{\alpha\beta} = O(r^{-3})
\end{eqnarray*}
---see e.g. \cite{Sim84b}. Thus, to leading order any stationary
electrovacuum spacetime is asymptotically a Kerr-Newman spacetime. }
\end{remark}

In \cite{ColVal16a} the following result has been proved:

\begin{theorem}
\label{Theorem:KerrNewmanCharacterisation}
Let $(\mathcal{M},\bmg,\bmF)$ be a smooth electrovacuum spacetime
satisfying the matter alignment condition with a stationary
asymptotically flat end $\mathcal{M}_\infty$ generated by a Killing
spinor $\kappa_{AB}$. Let both the Komar mass associated to the
Killing vector $\xi_{AA'}\equiv \nabla^B{}_{A'}\kappa_{AB}$ and the
total electromagnetic charge in $\mathcal{M}_\infty$ be non-zero. Then
$(\mathcal{M},\bmg,\bmF)$ is locally isometric to the Kerr-Newman spacetime.  
\end{theorem}

The above result can be regarded as a refinement of the
characterisations of the Kerr-Newman spacetime given in \cite{Won09}. 

\section{The Killing spinor evolution system in electrovacuum
  spacetimes}
\label{Section:EvolutionKillingSpinors}

In this section we systematically investigate the interrelations
between the zero-quantities $H_{A'ABC}$, $S_{AA'BB'}$ and
$\Theta_{AB}$. The ultimate objective of this analysis is to obtain a
system of homogeneous wave equations for the zero-quantities.

\subsection{A wave equation for $\xi_{AA'}$}
Given a Killing spinor candidate $\kappa_{AB}$, the wave equation
\eqref{WaveEquationKappa} naturally implies a wave equation for the
Killing vector candidate $\xi_{AA'}$.  We first notice the following
alternative expression for the field $S_{AA'BB'}$:

\begin{lemma}
Let $\kappa_{AB}$ denote a symmetric spinor field in an electrovacuum
$(\mathcal{M},\bmg,\bmF)$. Then, one has that
\begin{equation}
S_{AA'BB'}=6 \bar{\phi}_{A'B'} \Theta_{AB} -  \frac{1}{2}
\nabla_{PA'}H_{B'AB}{}^{P} -  \frac{1}{2} \nabla_{PB'}H_{A'AB}{}^{P}.
\label{Identity:StoThetaDH}
\end{equation}
\end{lemma}

\begin{proof}
To obtain the identity one starts by substituting the expression $\xi_{AA'} =
\nabla^Q{}_{A'}\kappa_{QA}$ into the definition of $S_{AA'BB'}$,
equation \eqref{DefinitionS}. One then commutes covariant derivatives using
the commutators \eqref{BoxCommutators} and makes use of the
decompositions of $\nabla_{AA'}\kappa_{BC}$, $\nabla_{AA'}\xi_{BB'}$
and $S_{AA'BB'}$ given by equations \eqref{Decomposition:DerivativeKappa}, \eqref{DecompositionDerivativeKillingVectorCandidate} and \eqref{Symmetries:S},
respectively, to simplify.
\end{proof}

\begin{remark}
{\em Observe that in the above result the spinor
$\kappa_{AB}$ is not assumed to be a Killing spinor candidate. }
\end{remark}

The latter is used, in turn, to obtain the main result of this section:

\begin{lemma}
\label{Lemma:WaveEquationKillingVectorCandidate}
Let $\kappa_{AB}$ denote a Killing spinor candidate in an
electrovacuum spacetime $(\mathcal{M},\bmg,\bmF)$. Then the Killing
vector candidate $\xi_{AA'}\equiv \nabla^Q{}_{A'} \kappa_{AQ}$
satisfies the wave equation
\begin{equation}
\square \xi_{AA'}=-2 \xi^{PP'} \Phi_{APA'P'} + \Phi^{PQ}{}_{A'}{}^{P'} H_{P'APQ} -  \
\Psi_{APQD} H_{A'}{}^{PQD} + 6 \bar{\phi}_{A'}{}^{P'}
\nabla_{PP'}\Theta_{A}{}^{P}.
\label{WaveEquationKillingVectorCandidate}
\end{equation}
\end{lemma}

\begin{proof}
One makes use of the definition of $S_{AA'BB'}$ and the identity
\eqref{Identity:StoThetaDH} to write
\begin{eqnarray*}
&& \hspace{-1.5cm}\nabla^{AA'}\nabla_{AA'}\xi_{BB'} + \nabla^{AA'}\nabla_{BB'}\xi_{AA'} \
= 6 \Theta_{AB} \nabla^{AA'}\bar{\phi}_{A'B'} + 6 \bar{\phi}_{A'B'}
   \nabla^{AA'}\Theta_{AB} \\
&& \hspace{5cm}-  \frac{1}{2} \nabla^{AA'}\nabla_{CA'}H_{B'AB}{}^{C} -  \frac{1}{2} \nabla^{AA'}\nabla_{CB'}H_{A'AB}{}^{C}.
\end{eqnarray*}
The above expression can be simplified using the Maxwell
equations. Moreover, commuting covariant derivatives in the terms 
$\nabla^{AA'}\nabla_{CA'}H_{B'AB}{}^{C}$ and
$\nabla^{AA'}\nabla_{CB'}H_{A'AB}{}^{C}$ one arrives to
\begin{eqnarray*}
&& \square \xi_{AA'} = -2 \xi^{PP'} \Phi_{APA'P'} + \Phi^{PQ}{}_{A'}{}^{P'} H_{P'APQ} -  \
\Psi_{APQD} H_{A'}{}^{PQD} + 6 \bar{\phi}_{A'}{}^{P'}
   \nabla_{PP'}\Theta_{A}{}^{P} \\
&& \hspace{3cm}-  \nabla_{AA'}\nabla_{PP'}\xi^{PP'} -  \frac{1}{2} \nabla_{QA'}\nabla_{PP'}H^{P'}{}_{A}{}^{PQ}.
\end{eqnarray*}
Finally, using that $\xi_{AA'}$ is a Killing vector candidate (see
Remark 4) and that $\nabla^{AA'}H_{A'ABC}=0$ (see Remark 5) the result
follows. 
\end{proof}

\begin{remark}
{\em Important for the subsequent discussion is that the wave equation
\eqref{WaveEquationKillingVectorCandidate} takes, in tensorial terms,
the form}
\begin{equation}
\square \xi_a = -2 \Phi_{ab} \xi^b + J_a
\label{WaveEquationKillingVectorCandidate}
\end{equation}
where $J_a$ is defined in spinorial terms by
\[
J_{AA'}\equiv  \Phi^{PQ}{}_{A'}{}^{P'} H_{P'APQ} -  \
\Psi_{APQD} H_{A'}{}^{PQD} + 6 \bar{\phi}_{A'}{}^{P'}
\nabla_{PP'}\Theta_{A}{}^{P}. 
\]
In terms of the zero-quantity $\zeta_{AA'}$ to be introduced in
equation \eqref{Definition:zeta} one has 
\[
J_{AA'}\equiv  \Phi^{PQ}{}_{A'}{}^{P'} H_{P'APQ} -  \
\Psi_{APQD} H_{A'}{}^{PQD} - 6 \bar{\phi}_{A'}{}^{P'}
\zeta_{AP'}. 
\] 
\end{remark}
Thus, $J_{AA'}$ is an homogeneous expression of zero-quantities and
does not involve their derivatives.

\subsection{A wave equation for $H_{A'ABC}$}

\begin{lemma}
Let $\kappa_{AB}$ denote a Killing spinor candidate in an
electrovacuum spacetime $(\mathcal{M},\bmg,\bmF)$. Then the
zero-quantity $H_{A'ABC}$ satisfies the wave equation
\begin{eqnarray}
&& \square H_{A'BCD}=2 \Psi_{CDAF} H_{A'B}{}^{AF} + 2 \Psi_{BDAF}
   H_{A'C}{}^{AF} + 4 \phi_{D}{}^{A} \bar{\phi}_{A'}{}^{B'} H_{B'BCA}
   \nonumber \\
&& \hspace{4cm} -
   12 \bar{\phi}_{A'}{}^{B'} \nabla_{DB'}\Theta_{BC}- 2\nabla_D{}^{B'} S_{(BC)(A'B')}. \label{WaveEquation:H}
 \end{eqnarray}
\end{lemma}

\begin{proof}
We consider, again, the identity \eqref{Identity:StoThetaDH} in the
form
\[
\nabla_{AB'} H_{A'BC}{}^A = 6 \bar{\phi}_{A'B'} \Theta_{BC} -
S_{(BC)(A'B')}. 
\]
Applying the derivative $\nabla_D{}^{B'}$ to the above expression one
readily finds that
\[
\nabla_{D}{}^{B'}\nabla_{AB'}H_{A'BC}{}^{A} = 6 (\Theta_{BC} \nabla_{D}{}^{B'}\bar{\phi}_{A'B'} + \bar{\phi}_{A'B'} \nabla_{D}{}^{B'}\Theta_{BC}) -  \nabla_{D}{}^{B'}S_{(BC)(A'B')}).
\]
Using the identity \eqref{SecondDerivativeDecomposition} and the
box commutators \eqref{BoxCommutators} one obtains, after using the
Maxwell equations to simplify, the desired equation.
\end{proof}

\begin{remark}
{\em Observe that the righthand side of the wave
equation \eqref{WaveEquation:H} is an homogeneous expression in the
zero-quantity $H_{A'ABC}$ and the first order derivatives of
$\Theta_{AB}$ and $S_{AA'BB'}$. }
\end{remark}

\subsection{A wave equation for $\Theta_{AB}$}
In order to compute a wave equation for the zero-quantity associated
to the matter alignment condition it is convenient to introduce a
further zero-quantity:
\begin{equation}
\zeta_{AA'} \equiv \nabla^Q{}_{A'} \Theta_{AQ}.
\label{Definition:zeta}
\end{equation}
Clearly, if the matter alignment condition
\eqref{MatterAlignmentCondition} is satisfied, then $\zeta_{AA'}=0$. 
The reason for introducing this further field will become clear in
the sequel. Using the above definition one obtains the following:

\begin{lemma}
Let $\kappa_{AB}$ denote a symmetric spinor field in an electrovacuum spacetime
$(\mathcal{M},\bmg,\bmF)$. Then, one has that
\begin{equation}
\square \Theta_{AB} = 2 \Psi_{ABPQ}\Theta^{PQ} - 2
\nabla_B{}^{A'}\zeta_{AA'}. 
\label{WaveEquation:Theta}
\end{equation}
\end{lemma}

\begin{proof}
The wave equation follows from applying the derivative
$\nabla_B{}^{A'}$ to the definition of $\zeta_{AA'}$ and using the
identity \eqref{SecondDerivativeDecomposition} together with the box
commutators \eqref{BoxCommutators}. 
\end{proof}

\begin{remark}
{\em A direct computation using the definitions of
$\Theta_{AB}$ and $\zeta_{AA'}$ together with the expression for the
irreducible decomposition of $\nabla_{AA'}\kappa_{BC}$ given by
equation \eqref{Decomposition:DerivativeKappa} and the Maxwell
equations gives that
\begin{equation}
\zeta_{AA'} = - \nabla_{A'(A}\phi_{BC)} \kappa^{BC} +
\frac{4}{3}\xi^B{}_{A'} \phi_{AB} + \frac{1}{3}H_{A'ABC}\phi^{BC}.
\label{Identity:Zeta}
\end{equation}}
\end{remark}

\begin{remark}
{\em It follows directly from equation
\eqref{WaveEquation:Theta} that 
\[
\nabla^{AA'}\zeta_{AA'}=0.
\]
Alternatively, this property can be verified through a direct
computation using the identity \eqref{Identity:Zeta}.}
\end{remark} 

\medskip
As the right hand side of equation \eqref{WaveEquation:Theta} is an homogeneous expression in
$\Theta_{AB}$ and a first order derivative of $\zeta_{AA'}$, one needs
to construct a wave equation for $\zeta_{AA'}$. The required
expression follows from an involved computation ---as it can be seen
from the proof of the following lemma:

\begin{lemma}
Let $\kappa_{AB}$ denote a symmetric spinor field in an electrovacuum
$(\mathcal{M},\bmg,\bmF)$. Then, one has that
\begin{eqnarray}
&& \square \zeta_{AA'} =4 \zeta^{DB'} \phi_{AD} \bar{\phi}_{A'B'} + \frac{2}{3} \phi^{DB} 
\Psi_{DBCF} H_{A'A}{}^{CF} -  \frac{2}{3} \phi^{DB} \Psi_{ABCF} 
H_{A'D}{}^{CF} \nonumber\\
&& \hspace{2cm}-  \frac{4}{3} \phi_{A}{}^{D} \phi^{BC} 
\bar{\phi}_{A'}{}^{B'} H_{B'DBC} -  \frac{2}{3} H^{B'DBC}
\nabla_{AB'}\phi_{DA'BC} -  \frac{2}{3} H_{A'}{}^{DBC}
\nabla_{AB'}\phi_{D}{}^{B'}{}_{BC} \nonumber\\
&& \hspace{2cm}+ \frac{2}{3} \phi^{DB'BC} \nabla_{AB'}H_{A'DBC} +
\frac{2}{3} \phi^{D}{}_{A'}{}^{BC} \nabla_{AB'}H^{B'}{}_{DBC} -
\frac{4}{3} \phi^{DB}
\nabla_A{}^{B'}S_{(BD)(A'B')} \nonumber\\
&& \hspace{2cm} - 4 
\phi^{DB} \bar{\phi}_{A'}{}^{B'} \nabla_{BB'}\Theta_{AD} - \frac{2}{3}
\phi^{DB} \nabla_B{}^{B'}S_{(AD)(A'B')} +
\frac{2}{3} \nabla_{A}{}^{B'}\phi^{DB} \nabla_{CB'}H_{A'DB}{}^{C} \nonumber\\
&& \hspace{2cm}-  \frac{4}{3} \nabla_{A}{}^{B'}\phi^{DB}
   S_{(BD)(A'B')}. \label{WaveEquation:Zeta}
\end{eqnarray}
where $\phi_{AA'BC}\equiv\nabla_{AA'}\phi_{BC}$.
\end{lemma}

\begin{proof}
Consider the identity \eqref{Identity:Zeta} and apply the derivative
$\nabla^A{}_{B'}$ to obtain 
\begin{eqnarray*}
&& \nabla^{A}{}_{B'}\zeta_{AA'} = - \kappa^{BC}
\nabla^{A}{}_{B'}\nabla_{AA'}\phi_{BC} -  \nabla_{AA'}\phi_{BC}
\nabla^{A}{}_{B'}\kappa^{BC} + \frac{4}{3} (\phi_{AB}
\nabla^{A}{}_{B'}\xi^{B}{}_{A'} + \xi^{B}{}_{A'}
\nabla^{A}{}_{B'}\phi_{AB}) \\
&& \hspace{3cm} + \frac{1}{3} (H_{A'ABC} \nabla^{A}{}_{B'}\phi^{BC} + \phi^{BC} \nabla^{A}{}_{B'}H_{A'ABC}).
\end{eqnarray*}
Some further simplifications yield
\[
\nabla{}^A{}_{B'}\zeta_{AA'} = \frac{1}{3} \nabla^{A}{}_{B'}\phi^{BC} H_{A'ABC} + \frac{1}{3} \nabla^{A}{}_{A'}\phi^{BC} H_{B'ABC} -  \frac{1}{3} \phi^{AB} \nabla_{CB'}H_{A'AB}{}^{C} + \frac{2}{3} \phi^{AB} S_{(AB)(A'B')}.
\]
To obtain the required wave equation we apply $\nabla_D{}^{B'}$ to the
above expression and make use of the decomposition \eqref{SecondDerivativeDecomposition} on the
terms 
\[
\frac{1}{3} \nabla_D{}^{B'}\nabla^{A}{}_{B'}\phi^{BC} H_{A'ABC}, \qquad \nabla_D{}^{B'}\nabla{}^A{}_{B'}\zeta_{AA'}, \qquad -  \frac{1}{3} \phi^{AB} \nabla_D{}^{B'}\nabla_{CB'}H_{A'AB}{}^{C}.
\]
Finally, substitution of the wave equations for $\phi_{AB}$ and
$H_{A'ABCD}$, equations \eqref{WaveEquationMaxwellSpinor} and
\eqref{WaveEquation:H} yields the required expression homogeneous in
zero-quantities.
\end{proof}

\subsection{A wave equation for $S_{AA'BB'}$}
The discussion of the wave equation for the spinorial field
$S_{AA'BB'}$ is best carried out in tensorial notation. Accordingly,
let $S_{ab}$ denote the tensorial counterpart of the (not
necessarily Hermitian) spinor $S_{AA'BB'}$. Key to this computation
is the wave equation for the Killing vector candidate $\xi^a$,
equation \eqref{WaveEquationKillingVectorCandidate}. 

\begin{lemma}
Let $\kappa_{AB}$ denote a Killing spinor candidate in an
electrovacuum spacetime $(\mathcal{M},\bmg,\bmF)$. Then the
zero-quantity $S_{ab}$ satisfies the wave equation
\begin{equation}
\square S_{ab} = -2 \mathcal{L}_{\xi}T_{ab} + 2 
T_{b}{}^{c} S_{ac} + 2 T_{a}{}^{c} S_{bc} -  T^{cd} S_{cd}g_{ab} -
T_{ab} S^{c}{}_{c} - 2 C_{acbd} S^{cd} + \nabla_{a}J_{b} +
\nabla_{b}J_{a}
\label{WaveEquation:S}
\end{equation}
where $\mathcal{L}_{\xi}T_{ab}$ denotes the Lie derivative of the
energy-momentum of the Faraday tensor.  
\end{lemma}

\begin{proof}
The required expression follows from applying $\square = \nabla_a
\nabla^a$ to 
\[
S_{ab} = \nabla_a\xi_b + \nabla_b \xi_a,
\]
commuting covariant derivatives, using the wave equation
\eqref{WaveEquationKillingVectorCandidate}, the Einstein equation
\[
R_{ab}=T_{ab},
\]
the contracted Bianchi identity 
\[
\nabla^a C_{abcd} = \nabla_{[c} T_{d]b}
\] 
and the relation 
\[
\nabla_a \xi_b = \frac{1}{2}S_{ab} + \nabla_{[a}\xi_{b]}.
\]
\end{proof}

\medskip
A straightforward computation computation shows that the Lie
  derivative of the electromagnetic energy-momentum tensor can be
  expressed in terms of the Lie derivative of the Faraday tensor and
  the zero-quantity $S_{ab}$ as
\[
\mathcal{L}_{\xi}T_{ab} = - \frac{1}{4} F_{cd} F^{cd} S_{ab} -  F_{a}{}^{c} F_{b}{}^{d} S_{cd} + \frac{1}{2} F_{c}{}^{f} F^{cd} g_{ab} S_{df} + F_{b}{}^{c} \mathcal{L}_{\xi}F_{ac} + F_{a}{}^{c} \mathcal{L}_{\xi}F_{bc} -  \frac{1}{2} F^{cd} g_{ab} \mathcal{L}_{\xi}F_{cd}.
\]
Furthermore, the Lie derivative of the Faraday tensor can be expressed in terms of the Lie derivative of the Maxwell spinor as
\begin{align*}
\mathcal{L}_{\xi}F_{AA'BB'}=\left(\mathcal{L}_{\xi}\phi_{AB}-\frac{1}{2}S_{AC'BD'}\bar{\phi}^{C'D'}\right)\epsilon_{A'B'} + \text{complex conjugate},
\end{align*}
where the \emph{Lie derivative of the Maxwell spinor} is defined by
\begin{equation}
\mathcal{L}_{\xi}\phi_{AB}\equiv \xi^{CC'}\nabla_{CC'}\phi_{AB}+\phi_{C(A}\nabla_{B)C'}\xi^{CC'}
\label{Definition:LieDerivativeMaxwell}
\end{equation}
---see Section 6.6 in \cite{PenRin86}. This expression can be written
in terms of zero quantities by using the wave equations for the
Killing and Maxwell spinors, the Maxwell equations and the identity
\begin{equation*}
\kappa^{D}{}_{(A}\Psi_{B)DEF}\phi^{EF}=\frac{1}{2}\Psi_{ABCD}\Theta^{CD}+\frac{1}{3}\phi^{EF}\nabla_{(A|}{}^{A'}H_{A'|BEF)}
\end{equation*}
along with the wave equations for the Killing and Maxwell spinors and
the Maxwell equations, equations \eqref{WaveEquationKappa} and
\eqref{WaveEquationKillingVectorCandidate}, so as to obtain
\begin{equation*}
\mathcal{L}_{\xi}\phi_{AB}=-\frac{3}{2}\nabla_{(A}{}^{A'}\zeta_{B)A'}+H_{A'CD(A}\nabla_{B)}{}^{A'}\phi^{CD}-\phi^{CD}\nabla_{(A|}{}^{A'}H_{A'|BCD)}.
\end{equation*}

From the previous discussion it follows that:

\begin{lemma}
\label{Lemma:LieDerivativeEnergyMomentumTensor}
Let $\kappa_{AB}$ denote a Killing spinor candidate in an
electrovacuum spacetime $(\mathcal{M},\bmg,\bmF)$. Then the Lie
derivative $\mathcal{L}_\xi T_{ab}$ can be expressed as an homogeneous
expression in the zero-quantities
\[
S_{AA'BB'}, \quad \zeta_{AA'}, \quad H_{A'ABC}
\]
and their first order derivatives.
\end{lemma}

\begin{remark}
{\em In the context of the present discussion the object $\mathcal{L}_\xi
\phi_{AB}$, as defined in \eqref{Definition:LieDerivativeMaxwell}, must be regarded as a convenient shorthand for a
complicated expression. It is only consistent with the usual notion of
Lie derivative of tensor fields if $\xi^{AA'}$ is the spinorial
counterpart of a conformal Killing vector $\xi^a$ ---see
\cite{PenRin86}, Section 6.6, for further discussion on this point.}
\end{remark}

\subsection{Summary}
We summarise the discussion of the present section in the following:

\begin{proposition}
Let $\kappa_{AB}$ denote a Killing spinor candidate in an
electrovacuum spacetime $(\mathcal{M},\bmg,\bmF)$. Then the zero-quantities
\[
H_{A'ABC}, \quad \Theta_{AB}, \quad \zeta_{AA'}, \quad S_{AA'BB'}
\]
satisfy a system of wave equations, consisting of equations
\eqref{WaveEquation:H}, \eqref{WaveEquation:Theta},
\eqref{WaveEquation:Zeta} and \eqref{WaveEquation:S}, which is homogeneous on the above
zero-quantities and their first order derivatives. 
\end{proposition}

A direct consequence of the above and the uniqueness of solutions to
homogeneous wave equations is the following:

\begin{theorem}
\label{Theorem:SpacetimeConditionKillingSpinor}
Let $\kappa_{AB}$ denote a Killing spinor candidate in an
electrovacuum spacetime $(\mathcal{M},\bmg,\bmF)$ and let
$\mathcal{S}$ denote a Cauchy hypersurface of
$(\mathcal{M},\bmg,\bmF)$. The spinor $\kappa_{AB}$
is an actual Killing spinor if and only if on $\mathcal{S}$ one has that
\begin{subequations}
\begin{eqnarray}
H_{A'ABC}|_{\mathcal{S}}=0, &\qquad& \nabla_{EE'} H_{A'ABC}|_{\mathcal{S}}=0  \label{DataCondition1}\\
S_{AA'BB'}|_{\mathcal{S}}=0,&\qquad &\nabla_{EE'} S_{AA'BB'}|_{\mathcal{S}}=0 \label{DataCondition2}\\
\Theta_{AB}|_{\mathcal{S}}=0, &\qquad &\nabla_{EE'} \Theta_{AB}|_{\mathcal{S}}=0 \label{DataCondition3}\\
\zeta_{AA'}|_{\mathcal{S}}=0, &\qquad &\nabla_{EE'} \zeta_{AA'}|_{\mathcal{S}}=0.  \label{DataCondition4}
\end{eqnarray}
\end{subequations}
\end{theorem}

\begin{proof}
The initial data for the homogeneous system of wave equations for the
fields $H_{A'ABC}$, $\Theta_{AB}$, $\zeta_{AA'}$ and $S_{AA'BB'}$
given by equations \eqref{WaveEquation:H}, \eqref{WaveEquation:Theta},
\eqref{WaveEquation:Zeta} and \eqref{WaveEquation:S} 
consists of the values of these fields and their normal derivatives at
the Cauchy surface $\mathcal{S}$. Because of the homogeneity of
the equations, the unique solution to these
equations with vanishing initial data is given by
\[
H_{A'ABC}=0, \quad \Theta_{AB}=0, \quad \zeta_{AA'}=0, \quad S_{AA'BB'}=0.
\]
Thus, if this is the case, the spinor $\kappa_{AB}$ satisfies the Killing
equation on $\mathcal{M}$ and, accordingly, it is a Killing
spinor. Conversely, given a Killing spinor $\kappa_{AB}$ over
$\mathcal{M}$, its restriction to $\mathcal{S}$ satisfies the
conditions \eqref{DataCondition1}-\eqref{DataCondition4}. 
\end{proof}

\begin{remark}
{\em As the spinorial zero-fields $H_{A'ABC}$, $\Theta_{AB}$,
  $\zeta_{AA'}$ and $S_{AA'BB'}$ can be expressed in terms of the
  spinor $\kappa_{AB}$, it follows that the conditions
  \eqref{DataCondition1}-\eqref{DataCondition4} are, in fact, conditions on
  $\kappa_{AB}$, and its (spacetime) covariant derivatives up to third
  order. In the next section it will be shown how these conditions can
be expressed in terms of objects intrinsic to the hypersurface
$\mathcal{S}$.}
\end{remark}

\section{The Killing spinor data equations}
\label{Section:KillingSpinorDataEquation}

The purpose of this section is to show how the conditions
\eqref{DataCondition1}-\eqref{DataCondition4} of Theorem
\ref{Theorem:SpacetimeConditionKillingSpinor} can be reexpressed as
conditions which are intrinsic to the hypersurface $\mathcal{S}$. To
this end we make use of the space-spinor formalism outlined in
\cite{BaeVal10b} with some minor notational changes.

\subsection{The space-spinor formalism}

In what follows assume that the spacetime $(\mathcal{M},\bmg)$
obtained as the development of Cauchy initial data
$(\mathcal{S},\bmh,\bmK)$ can be covered by a congruence of timelike
curves with tangent vector $\tau^a$ satisfying the normalisation
condition $\tau_a \tau^a=2$ ---the reason for normalisation will be
clarified in the following ---see equation \eqref{ContractionEpsilons}.  Associated to the vector $\tau^a$ one has
the projector
 \[
 h_a{}^b \equiv \delta_a{}^b - \frac{1}{2}\tau_a\tau^b 
 \]
projecting tensors into the distribution $\langle {\bm \tau}
\rangle^\perp$ of hyperplanes orthogonal to $\tau^a$. 

\begin{remark}
\label{Remark:HypersurfaceOrthogonality}
{\em The congruence
of curves needs not to be hypersurface orthogonal ---however, for
convenience it will be assumed that the vector field $\tau^a$ is
orthogonal to the Cauchy hypersurface $\mathcal{S}$. }
\end{remark}

Now, let $\tau^{AA'}$ denote the spinorial counterpart of
the vector $\tau^a$ ---by definition one has that
\begin{equation}
\tau_{AA'} \tau^{AA'} = 2.
\label{NormalisationTau}
\end{equation}
Let $\{o^A,\iota^A\}$ denote a normalised spin-dyad satisfying
$o_A\iota^A=1$. In the following we restrict the
attention to spin-dyads such that
\begin{equation}
\tau^{AA'}= o^A\bar{o}^{A'} + \iota^A\bar{\iota}^{A'}.
\label{ExpansionTau}
\end{equation}
It follows then that
\begin{equation}
\tau_{AA'} \tau^{BA'} =\delta_A{}^B,
\label{ContractionEpsilons}
\end{equation}
consistent with the normalisation condition
\eqref{NormalisationTau}. As a consequence of this relation, the
spinor $\tau^{AA'}$ can be used to introduce a formalism in which all
primed indices in spinors and spinorial equations are replaced by
unprimed indices by suitable contractions with $\tau_A{}^{A'}$.

\begin{remark}
{\em The set of transformations on the dyad $\{o^A,\iota^A\}$
  preserving the expansion \eqref{ExpansionTau} is given by the group $SU(2,\mathbb{C})$.}
\end{remark}

\subsubsection{The Sen connection}
The \emph{space-spinor} counterpart of the spinorial covariant
derivative 
$\nabla_{AA'}$ is defined as 
\begin{equation}
\nabla_{AB} \equiv \tau_B{}^{A'}\nabla_{AA'}.
\label{Definition:SpaceSpinorNabla}
\end{equation}
The derivative operator $\nabla_{AB}$ can be decomposed in irreducible
terms as
\begin{equation}
\nabla_{AB} = \frac{1}{2}\epsilon_{AB}\mathcal{P} +\mathcal{D}_{AB}
\label{Decomposition:SpaceSpinorNabla}
\end{equation}
where
\[
\mathcal{P} \equiv \tau^{AA'}\nabla_{AA'} = \nabla_Q{}^Q, \qquad \mathcal{D}_{AB}
\equiv \tau_{(A}{}^{A'} \nabla_{B)A'} = \nabla_{(AB)}. 
\]
The operator $\mathcal{P}$ is the directional derivative of
$\nabla_{AA'}$ in the direction of $\tau^{AA'}$ while
$\mathcal{D}_{AB}$ corresponds to the so-called \emph{Sen connection
  of the covariant derivative $\nabla_{AA'}$ implied by
  $\tau^{AA'}$}.

\subsubsection{The acceleration and the extrinsic curvature}
 Of particular relevance in the subsequent discussion
is the decomposition of the covariant derivative of the spinor
$\tau_{BB'}$, namely $\nabla_{AA'} \tau_{BB'}$. A calculation readily
shows that the content of this derivative is encoded in the spinors
\[
K_{AB} \equiv \tau_B{}^{A'} \mathcal{P} \tau_{AA'} , \qquad
K_{ABCD}\equiv \tau_D{}^{C'} \mathcal{D}_{AB}\tau_{CC'}
\]
corresponding, respectively, to the spinorial counterparts of the
acceleration and the Weingarten tensor, expressed in tensorial terms
as
\[
K_a \equiv -\frac{1}{2} \tau^b \nabla_b \tau_a, \qquad K_{ab} \equiv -
h_a{}^c h_b{}^d \nabla_c \tau_d.
\]
 It can be readily verified
that
\begin{equation}
K_{AB}= K_{(AB)}, \qquad K_{ABCD}= K_{(AB)(CD)}.
\label{BasicSymmetryK}
\end{equation}
In the sequel it will be convenient to express $K_{ABCD}$ in terms of
its irreducible components. To this end define
\[
\Omega_{ABCD} \equiv K_{(ABCD)},\qquad \Omega_{AB}\equiv K_{(A}{}^Q{}_{B)Q}, \qquad K\equiv K_{AB}{}^{CD},
\]
so that one can define 
\[
K_{ABCD} = \Omega_{ABCD} -\frac{1}{2}\epsilon_{A(C}\Omega_{D)B}-\frac{1}{2}\epsilon_{B(C}\Omega_{D)A}- \frac{1}{3}\epsilon_{A(C} \epsilon_{D)B} K.
\]
If the vector field $\tau^a$ is hypersurface orthogonal, then one has
that $\Omega_{AB}=0$, and thus the Weingarten tensor satisfies the symmetry
$K_{ab}=K_{(ab)}$  so that it can be regarded as the extrinsic
curvature of the leaves of a foliation of the spacetime
$(\mathcal{M},\bmg)$. If this is the case, in addition to the second
symmetry in \eqref{BasicSymmetryK} one has that
\[
K_{ABCD} = K_{CDAB}. 
\]
In particular, $K_{ABCD}$ restricted to the hypersurface
$\mathcal{S}$ satisfies the above symmetry and one has
$\Omega_{AB}=0$ ---cfr. Remark \ref{Remark:HypersurfaceOrthogonality}. 

\medskip
In what follows denote by $D_{AB}=D_{(AB)}$ the spinorial counterpart of the
Levi-Civita connection of the metric $\bmh$ on $\mathcal{S}$. The Sen
connection $\mathcal{D}_{AB}$ and the Levi-Civita connection $D_{AB}$
are related to each other through the spinor $K_{ABCD}$. For example,
for a valence 1 spinor $\pi_A$ one has that 
\[
\mathcal{D}_{AB} \pi_C = D_{AB}\pi_C + \frac{1}{2}K_{ABC}{}^Q\pi_Q,
\]
with the obvious generalisations for higher order spinors.

\subsubsection{Hermitian conjugation}
Given a spinor $\pi_A$, its \emph{Hermitian conjugate} is defined as
\[
\widehat{\pi}_A \equiv \tau_A{}^{Q'} \bar{\pi}_{Q'}.
\]
This operation can be extended in the obvious way to higher valence
pairwise symmetric spinors. The operation of Hermitian conjugation
allows to introduce a notion of \emph{reality}. Given spinors
$\nu_{AB}=\nu_{(AB)}$ and $\xi_{ABCD}=\xi_{(AB)(CD)}$, we say that
they are \emph{real} if and only if 
\[
\widehat{\nu}_{AB}= -\nu_{AB}, \qquad \widehat{\xi}_{ABCD}= \xi_{ABCD}.
\]
If the spinors are real then it can be shown that there exist real
spatial 3-dimensional tensors $\nu_i$ and $\xi_{ij}$ such that
$\nu_{AB}$ and $\xi_{ABCD}$ are their spinorial counterparts. We also
note that 
\[
\nu_{AB} \widehat{\nu}^{AB} \geq 0, \qquad \xi_{ABCD}\widehat{\xi}^{ABCD}\geq 0
\] 
independently of whether $\nu_{AB}$ and $\xi_{ABCD}$ are real or not.

Finally, it is observed that while the Levi-Civita covariant
derivative $D_{AB}$ is real in the sense that
\[
\widehat{D_{AB}\pi_C} = - D_{AB}\hat{\pi}_C,
\]
the Sen connection $\mathcal{D}_{AB}$ is not. More precisely, one has
that
\[
\widehat{\mathcal{D}_{AB}\pi_C} =- \mathcal{D}_{AB}\widehat{\pi}_C +
\frac{1}{2}K_{ABC}{}^Q \hat{\pi}_Q.
\]

\subsubsection{Commutators}
The main analysis of this section will require a systematic use of the
commutators of the the covariant derivatives
$\mathcal{P}$ and $\mathcal{D}_{AB}$. In order to discuss these in a
convenient manner it is convenient to define the Hermitian conjugate
of the Penrose box operator $\square_{AB} \equiv
\nabla_{C'(A}\nabla_{B)}{}^{C'}$ in the natural manner as
\[
\widehat{\square}_{AB} \equiv \tau_A{}^{A'}\tau_B{}^{B'} \square_{A'B'}.
\]
From the definition of $\square_{A'B'}$ it follows that 
\[
\widehat{\square}_{AB} \pi_C =\tau_A{}^{A'} \tau_B{}^{B'} \Phi_{FCA'B'} \pi^F.
\]

In terms of $\square_{AB}$ and $\widehat{\square}_{AB}$, the
commutators of $\mathcal{P}$ and $\mathcal{D}_{AB}$ read
\begin{subequations}
\begin{eqnarray}
&& [\mathcal{P},\mathcal{D}_{AB}] = \widehat{\square}_{AB} -
   \square_{AB} -\frac{1}{2}K_{AB}\mathcal{P} +
   K^D{}_{(A}\mathcal{D}_{B)D} - K_{ABCD}\mathcal{D}^{CD}, \label{SpatialCommutator1}\\
&&
   [\mathcal{D}_{AB},\mathcal{D}_{CD}]=\frac{1}{2}\big(\epsilon_{A(C}\square_{D)B}+\epsilon_{B(C}\square_{D)A}\big)+\frac{1}{2}\big(\epsilon_{A(C}\widehat{\square}_{D)B}+\epsilon_{B(C}\widehat{\square}_{D)A}\big)
  \nonumber\\
&& \hspace{3cm}+\frac{1}{2}\big(K_{CDAB}\mathcal{P}-K_{ABCD}\mathcal{P}\big)+K_{CDF(A}\mathcal{D}_{B)}{}^{F}-K_{ABF(C}\mathcal{D}_{D)}{}^{F}.\label{SpatialCommutator2}
\end{eqnarray}
\end{subequations}

\begin{remark}
\label{Remark:TorsionSenConnection}
{\em  Observe that on the hypersurface $\mathcal{S}$ the commutator
  \eqref{SpatialCommutator2} involves only objects intrinsic to
  $\mathcal{S}$. Notice, also, that the Sen connection $\mathcal{D}_{AB}$
has torsion. Namely, for a scalar $\phi$ one has that 
\[
[\mathcal{D}_{AB},\mathcal{D}_{CD}] \phi =
K_{CDF(A}\mathcal{D}_{B)}{}^{F}\phi-K_{ABF(C}\mathcal{D}_{D)}{}^{F} \phi.
\]}
\end{remark}

\subsection{Basic decompositions}
The purpose of this section is to provide a systematic discussion of
the irreducible decompositions of the various spinorial fields and equations that
will be required in the subsequent analysis.

\subsubsection{Space-spinor decomposition of the Killing spinor and
  Maxwell equations}
For reference, we provide a brief discussion of the space-spinor
decomposition of the Killing equation, equation \eqref{KillingSpinorEquation}, and the
Maxwell equation, equation \eqref{EinsteinMaxwell2}. 

\medskip
Contracting the Killing spinor equation \eqref{KillingSpinorEquation}
in the form $\nabla_{(A|A'|}\kappa_{CD)}=0$ with $\tau_B{}^{A'}$ one
obtains
\[
\nabla_{(A|B|}\kappa_{CD)}=0
\]
where $\nabla_{AB}$ is the differential operator defined in equation
\eqref{Definition:SpaceSpinorNabla}. Using the decomposition
\eqref{Definition:SpaceSpinorNabla} one further obtains
\[
\frac{1}{2}\epsilon_{(A|B|} \mathcal{P} \kappa_{CD)} + \mathcal{D}_{(A|B|}\kappa_{CD)}=0.
\]
Taking, respectively, the trace and the totally symmetric part of the
above expression one readily obtains the equations
\begin{subequations}
\begin{eqnarray}
&& \mathcal{P} \kappa_{AB} + \mathcal{D}_{(A}{}^Q\kappa_{B)Q}=0, \label{KillingSpinorEquationDecomposition1}\\
&& \mathcal{D}_{(AB}\kappa_{CD)}=0. \label{KillingSpinorEquationDecomposition2}
\end{eqnarray}
\end{subequations}
Equation \eqref{KillingSpinorEquationDecomposition1} can be
naturally interpreted as an evolution equation for the spinor
$\kappa_{AB}$ while equation
\eqref{KillingSpinorEquationDecomposition2} plays the role of a
\emph{constraint}.

\medskip
A similar calculation applied to the Maxwell equation, equation
\eqref{EinsteinMaxwell2}, in the form $\nabla^A{}_{A'}\phi_{AC}=0$
yields the equations
\begin{subequations}
\begin{eqnarray}
&& \mathcal{P}\phi_{AB} - 2 \mathcal{D}^Q{}_{(A}\phi_{B)Q}=0, \label{MaxwellEquationDecomposition1}\\
&& \mathcal{D}^{AB}\phi_{AB}=0.\label{MaxwellEquationDecomposition2}
\end{eqnarray}
\end{subequations}

Again, equation \eqref{MaxwellEquationDecomposition1} is an
evolution equation for the Maxwell spinor $\phi_{AB}$ while
\eqref{MaxwellEquationDecomposition2} is the spinorial version
of the electromagnetic Gauss constraint. 

\begin{remark}
{\em The operation of Hermitian conjugation can be used to define, respectively, the
  \emph{electric} and \emph{magnetic} parts of the Maxwell spinor:}
\[
E_{AB} \equiv \frac{1}{2}\big(\widehat{\phi}_{AB}-\phi_{AB} \big),
\qquad B_{AB} \equiv \frac{\mbox{\em i}}{2}\big(
\phi_{AB} + \widehat{\phi}_{AB}\big).
\]
{\em It can be readily verified that}
\[
\widehat{E}_{AB}=- E_{AB}, \qquad \widehat{B}_{AB} = -B_{AB}.
\]
{\em Thus, $E_{AB}$ and $B_{AB}$ are the spinorial counterparts of
3-dimensional tensors $E_i$ and $B_i$ ---the electric and magnetic
parts of the Faraday tensor with respect to the normal to the
hypersurface $\mathcal{S}$. }
\end{remark}

\subsubsection{The decomposition of the components of the curvature}
Crucial for our subsequent discussion will be the fact that the
restriction of the Weyl spinor $\Psi_{ABCD}$ to an hypersurface
$\mathcal{S}$ can be expressed in terms of quantities intrinsic to the
hypersurface. 

In analogy to the case of the Maxwell spinor $\phi_{AB}$, the
Hermitian conjugation operation can be used to decompose the Weyl
spinor $\Psi_{ABCD}$ into its electric and magnetic parts with respect
to the normal to $\mathcal{S}$ as
\[
E_{ABCD}\equiv \frac{1}{2}\big( \Psi_{ABCD}+ \hat{\Psi}_{ABCD}\big),
\qquad B_{ABCD} \equiv \frac{\mbox{i}}{2}\big( \hat{\Psi}_{ABCD} - \Psi_{ABCD}\big)
\]
so that 
\[
\Psi_{ABCD} =E_{ABCD} + \mbox{i} B_{ABCD}. 
\]
The electrovacuum Bianchi identity \eqref{BianchiIdentity} implies on $\mathcal{S}$
the constraint
\[
\mathcal{D}^{AB}\Psi_{ABCD} =-2 \widehat{\phi}^{AB} \mathcal{D}_{AB}\phi_{CD}.
\]
Finally, using the Gauss-Codazzi and Codazzi-Mainardi equations one finds that
\begin{eqnarray*}
&& E_{ABCD} = -r_{(ABCD)} + \frac{1}{2}\Omega_{(AB}{}^{PQ}
   \Omega_{CD)PQ} - \frac{1}{6}\Omega_{ABCD} K + E_{(AB}E_{CD)}, \\
&& B_{ABCD}= -\mbox{i} D^Q{}_{(A}\Omega_{BCD)Q},
\end{eqnarray*}
where $r_{ABCD}$ is the spinorial counterpart of the Ricci tensor of
the intrinsic metric of the hypersurface $\mathcal{S}$. 

\subsubsection{Decomposition of the spatial derivatives of the Killing
  spinor candidate}
Given a spinor $\kappa_{AB}$ defined on the Cauchy hypersurface
$\mathcal{S}$, it will prove convenient to define:
\begin{subequations}
\begin{eqnarray}
&&\xi\equiv\mathcal{D}^{AB}\kappa_{AB}, \label{Definition:xi0}\\
&&\xi_{AB}\equiv\frac{3}{2}\mathcal{D}_{(A}{}^{C}\kappa_{B)C}, \label{Definition:xi2}\\
&&\xi_{ABCD}\equiv\mathcal{D}_{(AB}\kappa_{CD)}. \label{Definition:xi4}
\end{eqnarray}
\end{subequations}
These spinors correspond to the irreducible components of the Sen derivative of $\kappa_{AB}$, as follows:
\[
\mathcal{D}_{AB}\kappa_{CD}=\xi_{ABCD}-\frac{1}{3}\epsilon_{A(C}\xi_{D)B}-\frac{1}{3}\epsilon_{B(C}\xi_{D)A}-\frac{1}{3}\epsilon_{A(C}\epsilon_{D)B}\xi.
\]

Using the commutation relation for the Sen derivatives, equation
\eqref{SpatialCommutator2}, we can also calculate the derivatives of
$\xi$ and $\xi_{AB}$. The irreducible components of
$\mathcal{D}_{AB}\xi_{CD}$ are given on $\mathcal{S}$ ---where $\Omega_{AB}=0$--- by
\begin{subequations}
\begin{align}
\mathcal{D}_{AB}\xi^{AB}&=-\frac{1}{2}K\xi+\frac{3}{4}\Omega^{ABCD}\xi_{ABCD}+\frac{3}{2}\Theta_{AB}\widehat{\phi}^{AB},\label{totalcontract} \\
\mathcal{D}_{A(B}\xi_{C)}{}^{A}&=-\mathcal{D}_{BC}\xi-\frac{3}{2}\Psi_{BCAD}\kappa^{AD}+\frac{2}{3}K\xi_{BC}+\frac{1}{2}\Omega_{BCAD}\xi^{AD}-\frac{3}{2}\Omega_{(B}{}^{ADF}\xi_{C)ADF}\notag \\
&\qquad+\frac{3}{2}\mathcal{D}_{AD}\xi_{BC}{}^{AD}-3\Theta_{A(B}\hat{\phi}_{C)}{}^{A},
  \label{symcontract}\\
\mathcal{D}_{(AB}\xi_{CD)}&=3\Psi_{F(ABC}\kappa_{D)}{}^{F}+K\xi_{ABCD}-\frac{1}{2}\xi\Omega_{ABCD}+\Omega_{(ABC}{}^{F}\xi_{D)F}\notag \\
&\qquad-\frac{3}{2}\Omega_{(AB}{}^{PQ}\xi_{CD)PQ}+3\mathcal{D}^{F}{}_{(A}\xi_{BCD)F}-3\Theta_{(AB}\widehat{\phi}_{CD)}\label{totalsym},
\end{align}
\end{subequations}
where we have also used the \emph{Hermitian conjugate} of the Maxwell spinor, defined by
\begin{equation*}
\widehat{\phi}_{AB}\equiv\tau_{A}{}^{A'}\tau_{B}{}^{B'}\bar{\phi}_{A'B'}.
\end{equation*}
Note that in \eqref{symcontract}, the term $\mathcal{D}_{AB}\xi$
appears ---there is no independent equation for the Sen derivative of
$\xi$.

We can also calculate the second order derivatives of $\xi$. Again, on
the hypersurface $\mathcal{S}$ these take the form:
\begingroup
\allowdisplaybreaks 
\begin{subequations}
\begin{align}
\mathcal{D}_{AB}\mathcal{D}^{AB}\xi =\;& -  \frac{1}{6} K^2 \xi+\frac{1}{2} K\widehat{\phi}^{AB} \Theta_{AB} - 2 \widehat{\phi}^{AB} \phi_{AB} \xi + \frac{2}{3} \xi^{AB} \mathcal{D}_{AB}K + 3 \Theta^{AB} \mathcal{D}_{BC}\widehat{\phi}_{A}{}^{C} \notag\\
&  - 4 \widehat{\phi}^{AB} \phi_{A}{}^{C} \xi_{BC} -  \frac{3}{2} \Psi^{ABCD} \xi_{ABCD} + 3 \widehat{\phi}^{AB} \phi^{CD} \xi_{ABCD} - 3 \widehat{\phi}^{AB} \Theta^{CD} \Omega_{ABCD} \notag\\
& -  \frac{1}{2} \Omega_{ABCD} \Omega^{ABCD} \xi + \frac{5}{4} K \Omega^{ABCD} \xi_{ABCD} + 3 \kappa^{AB} \Psi_{A}{}^{CDF} \Omega_{BCDF} \notag\\
&  -  \frac{3}{2} \Omega_{AB}{}^{FG} \Omega^{ABCD} \xi_{CDFG} - 3 \kappa^{AB} \widehat{\phi}^{CD} \mathcal{D}_{BD}\phi_{AC} + 3 \kappa^{AB} \widehat{\phi}_{A}{}^{C} \mathcal{D}_{CD}\phi_{B}{}^{D} \notag\\
& -  \frac{3}{2} \kappa^{AB} \mathcal{D}_{CD}\Psi_{AB}{}^{CD}  + \frac{1}{2} \xi^{AB} \mathcal{D}_{CD}\Omega_{AB}{}^{CD} + \frac{3}{2} \mathcal{D}_{CD}\mathcal{D}_{AB}\xi^{ABCD} \notag\\
& + \frac{3}{2} \xi^{ABCD} \mathcal{D}_{DF}\Omega_{ABC}{}^{F} -  \frac{9}{2} \Omega^{ABCD} \mathcal{D}_{DF}\xi_{ABC}{}^{F}, \label{contractxi}\\
\mathcal{D}^{C}{}_{(A}\mathcal{D}_{B)C}\xi =&\;\frac{1}{2}\Omega_{ABCD}\mathcal{D}^{CD}\xi -\frac{1}{3}K\mathcal{D}_{AB}\xi, \label{symcontractxi}\\
\mathcal{D}_{(AB}\mathcal{D}_{CD)}\xi =&\; \frac{1}{3} \widehat{\phi}^{EF} \Theta_{EF} \Omega_{ABCD} -  \Psi_{ABCD} \xi -  \frac{5}{9} K \Omega_{ABCD} \xi + \frac{1}{6} \Omega^{EFPQ} \Omega_{ABCD} \xi_{EFPQ}\notag \\
& + \frac{8}{9} K^2 \xi_{ABCD} + \frac{1}{3} \xi \mathcal{D}_{E(A}\Omega_{BCD)}{}^{E} -  \frac{10}{3} K \mathcal{D}_{E(A}\xi_{BCD)}{}^{E} \notag\\
& + \frac{3}{2} \mathcal{D}_{(AB}\mathcal{D}_{|EF|}\xi_{CD)}{}^{EF} + \frac{3}{2} \mathcal{D}_{F(A}\mathcal{D}_{B|E|}\xi_{CD)}{}^{EF} + \frac{1}{2} \mathcal{D}_{(A|F}\mathcal{D}_{E|}{}^{F}\xi_{BCD)}{}^{E} \notag\\
& + \frac{8}{3} K \kappa_{(A}{}^{E}\Psi_{BCD)E} - \frac{3}{2} \kappa_{(A}{}^{E}\mathcal{D}_{B|F|}\Psi_{CD)E}{}^{F} -  \frac{3}{2} \kappa^{EF}\mathcal{D}_{(AB}\Psi_{CD)EF} \notag\\
& -  \frac{1}{2} \kappa^{EF}\mathcal{D}_{F(A}\Psi_{BCD)E} + 2 \xi \widehat{\phi}_{(AB}\phi_{CD)} -  \frac{8}{3} K \widehat{\phi}_{(AB}\Theta_{CD)} + \Theta_{(AB}\mathcal{D}_{C|E|}\widehat{\phi}_{D)}{}^{E} \notag\\
& + 3 \Theta_{(A}{}^{E}\mathcal{D}_{BC}\widehat{\phi}_{D)E} + \Theta_{(A}{}^{E}\mathcal{D}_{B|E|}\widehat{\phi}_{CD)} + 2 \Psi_{E(ABC}\xi_{D)}{}^{E} + \frac{1}{6} \xi \Omega_{(AB}{}^{EF}\Omega_{CD)EF} \notag\\
& -  \frac{14}{9} K \Omega_{E(ABC}\xi_{D)}{}^{E} -  \frac{5}{3} K \Omega_{(AB}{}^{EF}\xi_{CD)EF} + \frac{2}{3} \Omega_{E(ABC}\mathcal{D}_{D)}{}^{E}\xi \notag\\
& + \frac{3}{2} \Omega_{(ABC}{}^{E}\mathcal{D}^{FP}\xi_{D)EFP} - \Omega_{(AB}{}^{EF}\mathcal{D}_{C}{}^{P}\xi_{D)EFP} + \frac{1}{2} \Omega_{(AB}{}^{EF}\mathcal{D}_{|FP|}\xi_{CD)E}{}^{P} \notag\\
& -  \frac{3}{2} \Omega_{(A}{}^{EFP}\mathcal{D}_{BC}\xi_{D)EFP} + \frac{1}{2} \Omega_{(A}{}^{EFP}\mathcal{D}_{B|P|}\xi_{CD)EF} + \frac{2}{3} \xi_{(AB}\mathcal{D}_{CD)}K \notag\\
& + \frac{1}{2} \xi_{(A}{}^{E}\mathcal{D}_{B|F|}\Omega_{CD)E}{}^{F} + \frac{1}{2} \xi^{EF}\mathcal{D}_{(AB}\Omega_{CD)EF} + \frac{1}{6} \xi^{EF}\mathcal{D}_{F(A}\Omega_{BCD)E} \notag\\
& + \frac{2}{3} \xi_{E(ABC}\mathcal{D}_{D)}{}^{E}K + \frac{1}{2} \xi_{(AB}{}^{EF}\mathcal{D}_{C|P|}\Omega_{D)EF}{}^{P} + \frac{3}{2} \xi_{(A}{}^{EFP}\mathcal{D}_{BC}\Omega_{D)EFP} \notag\\
& + \frac{1}{2} \xi_{(A}{}^{EFP}\mathcal{D}_{B|P|}\Omega_{CD)EF} + \kappa_{(A}{}^{E}\widehat{\phi}_{BC}\mathcal{D}_{D)F}\phi_{E}{}^{F} - 3 \kappa_{(A}{}^{E}\widehat{\phi}_{B}{}^{F}\mathcal{D}_{CD)}\phi_{EF} \notag\\
& + 2 \kappa_{(A}{}^{E}\widehat{\phi}_{B}{}^{F}\mathcal{D}_{C|F|}\phi_{D)E} + \kappa^{EF}\widehat{\phi}_{(AB}\mathcal{D}_{C|F|}\phi_{D)E} + 3 \kappa^{EF}\widehat{\phi}_{E(A}\mathcal{D}_{BC}\phi_{D)F} \notag\\
& +  \frac{1}{2} \kappa_{E(A}\Psi_{B}{}^{EFP}\Omega_{CD)FP} -  \frac{1}{2} \kappa^{EF}\Psi^{P}{}_{E(AB}\Omega_{CD)FP} + \frac{3}{2} \kappa^{EF}\Psi^{P}{}_{EF(A}\Omega_{BCD)P} \notag\\
& + \frac{10}{3} \widehat{\phi}_{(AB}\phi_{C}{}^{E}\xi_{D)E} + \frac{2}{3} \widehat{\phi}_{(A}{}^{E}\phi_{BC}\xi_{D)E} + \frac{2}{3} \phi_{E(A}\widehat{\phi}_{B}{}^{E}\xi_{CD)} -  \widehat{\phi}_{(AB}\phi^{EF}\xi_{CD)EF} \notag\\
& + \widehat{\phi}_{(A}{}^{E}\phi_{B}{}^{F}\xi_{CD)EF} + 3 \phi_{E}{}^{F}\widehat{\phi}_{(A}{}^{E}\xi_{BCD)F} + \frac{1}{6} \widehat{\phi}_{(AB}\Theta^{EF}\Omega_{CD)EF} \notag\\
& + \frac{2}{3} \widehat{\phi}_{(A}{}^{E}\Theta_{B}{}^{F}\Omega_{CD)EF} + \frac{3}{2} \Theta_{E}{}^{F}\widehat{\phi}_{(A}{}^{E}\Omega_{BCD)F} + \frac{1}{6} \widehat{\phi}^{EF}\Theta_{(AB}\Omega_{CD)EF} \notag\\
& + \frac{3}{2} \widehat{\phi}^{EF}\Theta_{E(A}\Omega_{BCD)F} + \frac{1}{2} \Omega_{(ABC}{}^{P}\Omega_{D)EFP}\xi^{EF} + \frac{1}{6} \Omega_{EFP(A}\Omega_{BC}{}^{FP}\xi_{D)}{}^{E}  \notag\\
& - \frac{3}{4} \Omega_{(ABC}{}^{E}\Omega_{D)}{}^{FPQ}\xi_{EFPQ} -  \frac{3}{4} \Omega_{E}{}^{FPQ}\Omega_{(ABC}{}^{E}\xi_{D)FPQ}  + \frac{1}{12} \Omega_{(AB}{}^{EF}\Omega_{CD)}{}^{PQ}\xi_{EFPQ} \notag\\
& + \frac{1}{3} \Omega_{E(A}{}^{PQ}\Omega_{BC}{}^{EF}\xi_{D)FPQ} + \frac{1}{12} \Omega_{EF}{}^{PQ}\Omega_{(AB}{}^{EF}\xi_{CD)PQ} \label{symxi}.
\end{align}
\end{subequations}
\endgroup

\begin{remark}
{\em It is of interest to remark that equation \eqref{symcontractxi} is
just the statement that the Sen connection has torsion  ---cf. Remark
\ref{Remark:TorsionSenConnection}. }
\end{remark}

An important and direct consequence of the above expressions is the following:

\begin{lemma}
\label{Lemma:DerivativesKappa}
Assume that $\Omega_{AB}=0$ and $\mathcal{D}_{(AB}\kappa_{CD)}=0$ on $\mathcal{S}$. Then
\[
\mathcal{D}_{AB}\mathcal{D}_{CD}\mathcal{D}_{EF}\kappa_{GH}=H_{ABCDEFGH}
\]
on $\mathcal{S}$, where $H_{ABCDEFGH}$ is a linear combination of $\kappa_{AB}$, $\mathcal{D}_{AB}\kappa_{CD}$ and $\mathcal{D}_{AB}\mathcal{D}_{CD}\kappa_{EF}$ with coefficients
depending on $\Psi_{ABCD}$, $K_{ABCD}$, $\phi_{AB}$,
$\widehat{\phi}_{AB}$ and $\mathcal{D}_{AB}\phi_{CD}$.
\end{lemma}

\begin{proof}
The proof of the above result follows from direct inspection of
equations \eqref{totalcontract}-\eqref{totalsym} and
\eqref{contractxi}-\eqref{symxi}. 
\end{proof}

\begin{remark}
{\em We observe that the above result is strictly not true if
  $\xi_{ABCD}=\mathcal{D}_{(AB}\kappa_{CD)}\neq 0$. }
\end{remark}

\subsection{The decomposition of the Killing spinor data equations}

In this section we provide a systematic discussion of the
decomposition of the \emph{Killing initial data conditions} in Theorem
\ref{Theorem:SpacetimeConditionKillingSpinor}. The main purpose of
this decomposition is to untangle the interrelations between the
various conditions and to obtain a \emph{minimal} set of equations
which is intrinsic to the Cauchy hypersurface $\mathcal{S}$.

\medskip
For the ease of the discussion we make explicit the assumptions we
assume to hold throughout this section:

\begin{assumption}
Given a Cauchy hypersurface $\mathcal{S}$ of an electrovacuum
spacetime $(\mathcal{M},\bmg)$, we assume that the hypothesis and
conclusions of Theorem \ref{Theorem:SpacetimeConditionKillingSpinor}
hold. 
\end{assumption}

Also, to ease the calculations, without loss of generality we assume:

\begin{assumption}
The spinor $\tau^{AA'}$ which on $\mathcal{S}$ is normal to
$\mathcal{S}$ is extended off the initial hypersurface in such a way
that it is the spinorial counterpart of the tangent vector to a congruence of
$\bmg$-geodesics. Accordingly one has that $K_{AB}=0$ ---that is, the
acceleration vanishes.
\end{assumption}

\subsubsection{Decomposing $H_{A'ABC}=0$}
Splitting the expression $\tau_{D}{}^{A'}H_{A'ABC}$ into irreducible
parts, and using the definitions
\eqref{Definition:xi0}-\eqref{Definition:xi4} gives that the condition
$H_{A'ABC}=0$ is equivalent to
\begin{subequations}
\begin{eqnarray}
&&\xi_{ABCD}=0, \label{Hdecompose1}\\
&&\mathcal{P}\kappa_{AB}=-\displaystyle\frac{2}{3}\xi_{AB}. \label{Hdecompose2}
\end{eqnarray}
\end{subequations}

Equation \eqref{Hdecompose1} is a condition intrinsic to the
hypersurface while \eqref{Hdecompose2} is extrinsic ---i.e. it
involves derivatives in the direction normal to $\mathcal{S}$.

\begin{remark}
{\em Observe that the conditions \eqref{Hdecompose1} and
  \eqref{Hdecompose2} are essentially the equations
  \eqref{KillingSpinorEquationDecomposition1} and
  \eqref{KillingSpinorEquationDecomposition2}. }
\end{remark}

\subsubsection{Decomposing $\nabla_{EE'} H_{A'ABC}=0$}

If $H_{A'ABC}=0$ on $\mathcal{S}$, it readily follows that
$\mathcal{D}_{EF}H_{A'ABC}=0$ on $\mathcal{S}$. Thus, in order
investigate the consequences of the second condition in \eqref{DataCondition1} it is
only necessary to consider the transverse derivative
$\mathcal{P}H_{A'ABC}$. It follows that
\[
\tau_{D}{}^{A'}\mathcal{P} H_{A'ABC}=\mathcal{P}\big(\tau_{D}{}^{A'}H_{A'ABC}\big)-H_{A'ABC}\mathcal{P}\tau_{D}{}^{A'}
\]
and so as $H_{A'ABC}|_{\mathcal{S}}=0$, the irreducible parts of
$\tau_{D}{}^{A'}\mathcal{P} H_{A'ABC}=0$ are given by
\begin{subequations}
\begin{align}
&\mathcal{P}\xi_{ABCD}=0, \label{PH1} \\
&\mathcal{P}^{2}\kappa_{AB}=-\frac{2}{3}\mathcal{P}\xi_{AB}.\label{PH2}
\end{align}
\end{subequations}

Taking equation \eqref{PH1} and commuting the $\mathcal{D}_{AB}$ and
$\mathcal{P}$ derivatives, and using equations \eqref{Hdecompose1} and
\eqref{Hdecompose2}, gives
\begin{align*}
\mathcal{P}\xi_{ABCD}&=\mathcal{P}\mathcal{D}_{(AB}\kappa_{CD)} \\
&=2\Psi^{F}{}_{(ABC}\kappa_{D)F}-\frac{1}{3}\xi\Omega_{ABCD}+\frac{2}{3}\Omega^{F}{}_{(ABC}\xi_{D)F} -\frac{2}{3}\mathcal{D}_{(AB}\xi_{CD)}-2\Theta_{(AB}\widehat{\phi}_{CD)}.
\end{align*}
Substituting for the derivative of $\xi_{AB}$ using \eqref{totalsym},
and using equations \eqref{Hdecompose1} and \eqref{Hdecompose2} again,
gives
\begin{equation}
\mathcal{P}\xi_{ABCD}=4\Psi^{F}{}_{(ABC}\kappa_{D)F}=0. \label{gradHdecompose}
\end{equation}

To reexpress condition \eqref{PH2}, we use the following result which
is obtained by commuting the $\mathcal{D}_{AB}$ and $\mathcal{P}$
derivatives:
\begin{align}
\mathcal{P}\xi_{AB} &= \frac{3}{2} \kappa^{CD} \Psi_{ABCD} -  3  \Theta_ {C(A}\widehat{\phi}_{B)}{}^{C}  -  \frac{1}{3} K \xi_{AB} + \frac{1}{2} \Omega_{ABCD} \xi^{CD} -  \frac{3}{2} \mathcal{D}_{C(A}\mathcal{P}\kappa_{B)}{}^{C}\label{Normxi2}.
\end{align}
Recall that the Killing spinor candidate $\kappa_{AB}$ satisfies the
homogeneous wave equation \eqref{WaveEquationKappa}. We can use the
space-spinor decomposition to split the wave operator into Sen
 and normal  derivative
operators. The result is:
\begin{align*}
\mathcal{P}^2\kappa_{AB} =& -2 \kappa^{CD} \Psi_{ABCD} + \frac{1}{3} K_{AB} \xi + \frac{2}{3} \Omega_{AB} \xi -  \tfrac{2}{3} K_{(A}{}^{C} \xi_{B)C} \\
&  -  \frac{4}{3} \Omega_{(A}{}^{C} \xi_{B)C} + K^{CD} \xi_{ABCD} + 2 \Omega^{CD} \xi_{ABCD} \\
& -  K \mathcal{P}\kappa_{AB} -  \frac{2}{3} \mathcal{D}_{AB}\xi +  \frac{4}{3} \mathcal{D}_{(A}{}^{C}\xi_{B)C} - 2 \mathcal{D}_{CD}\xi_{AB}{}^{CD}
\end{align*}
Applying conditions \eqref{Hdecompose1} and \eqref{Hdecompose2} to the
right hand side of the latter, evaluating at $\mathcal{S}$ (where
$\Omega_{AB}=0$) and setting $K_{AB}=0$ gives
\[
\mathcal{P}^2\kappa_{AB} = -2 \kappa^{CD} \Psi_{ABCD} + \frac{2}{3} K
\xi_{AB} -  \frac{2}{3} \mathcal{D}_{AB}\xi +  \frac{4}{3}
\mathcal{D}_{(A}{}^{C}\xi_{B)C}. 
\]
Then, using equations \eqref{Normxi2} and \eqref{symcontract}, as well
as \eqref{Hdecompose1} and \eqref{Hdecompose2} as needed, it can be shown that
\begin{equation}
\mathcal{P}^2\kappa_{AB}=-\frac{2}{3}\mathcal{P}\xi_{AB} \label{gradHdecompose2}
\end{equation}
which is exactly the condition we needed. Thus, we have shown that the
condition \eqref{PH2} is purely a consequence of the evolution equation for
the Killing spinor candidate, along with the conditions arising from
$H_{A'ABC}|_{\mathcal{S}}=0$. 

\medskip
In summary, if $\kappa_{AB}$ satisfies
$\square\kappa_{AB}+\Psi_{ABCD}\kappa^{CD}=0$, then:
\[
H_{A'ABC}|_{\mathcal{S}}=\mathcal{P}
H_{A'ABC}|_{\mathcal{S}}=0\,\Longleftrightarrow\, \xi_{ABCD}=0, \quad
\mathcal{P}\kappa_{AB}+\frac{2}{3}\xi_{AB}=0, \quad \Psi^{F}{}_{(ABC}\kappa_{D)F}=0.
\]

\subsubsection{Decomposing $\Theta_{AB}=0$}
As $\Theta_{AB}$ has no unprimed indices, it is already in a
space-spinor compatible form ---we have the condition:
\begin{equation}
\Theta_{AB}=\kappa_{(A}{}^{C}\phi_{B)C}=0. \label{Z=0}
\end{equation}

\subsubsection{Decomposing $\nabla_{EE'} \Theta_{AB}=0$}
If $\Theta_{AB}|_{\mathcal{S}}=0$, one only needs to consider the
normalderivative $\mathcal{P}\Theta_{AB}$. Using the evolution
equation for the spinor $\phi_{AB}$ implied by Maxwell equations,
equation \eqref{MaxwellEquationDecomposition1}, along with
\eqref{Hdecompose2} in the condition $\mathcal{P} \Theta_{AB}=0$ gives the
spatially intrinsic condition
\begin{equation}
\kappa_{(A|}{}^{C}\mathcal{D}_{CD}\phi_{|B)}{}^{D}=\frac{1}{3}\phi_{(A}{}^{C}\xi_{B)C} \label{Zdecomposecondition}
\end{equation}

\medskip
In summary, assuming \eqref{Hdecompose2} holds, then:
\begin{equation*}
\Theta_{AB}|_{\mathcal{S}}=\mathcal{P}
\Theta_{AB}|_{\mathcal{S}}=0\,\Longleftrightarrow\,\kappa_{(A}{}^{C}\phi_{B)C}=0,\quad
\kappa_{(A|}{}^{C}\mathcal{D}_{CD}\phi_{|B)}{}^{D}=\frac{1}{3}\phi_{(A}{}^{C}\xi_{B)C}.
\end{equation*}

\subsubsection{Decomposing $S_{AA'BB'}=0$}
Our point of departure to decompose the condition
$S_{AA'BB'}|_{\mathcal{S}}=0$ is the relation linking $S_{AA'BB'}$ to
$\Theta_{AB}$ and the derivative of $H_{A'ABC}$ given by equation
\eqref{Identity:StoThetaDH}. Splitting the derivative of $H_{A'ABC}$ into normal and tangential parts gives
\begin{equation}
S_{AA'BB'}=-6\bar{\phi}_{A'B'}\Theta_{AB}+\frac{1}{2}\tau^{C}{}_{(A'}\mathcal{P}
H_{B')ABC}+\tau_{D(A'}\mathcal{D}^{DC}H_{B')ABC}. 
\label{SDecompositionIntermmediate}
\end{equation}
We already have conditions ensuring that
$\Theta_{AB}|_{\mathcal{S}}=H_{A'ABC}|_{\mathcal{S}}=\mathcal{P}
H_{A'ABC}|_{\mathcal{S}}=0$, and so as a consequence we automatically
have that $S_{AA'BB'}|_{\mathcal{S}}=0$.

\subsubsection{Decomposing $\nabla_{EE'} S_{AA'BB'}=0$}
Again as $S_{AA'BB'}|_{\mathcal{S}}=0$, one only needs to consider the
normal derivative $\mathcal{P}S_{AA'BB'}$. Taking the normal
derivative of equation \eqref{SDecompositionIntermmediate} and using
that one has a Gaussian gauge gives on $\mathcal{S}$ that
\begin{align*}
\mathcal{P} S_{AA'BB'}=&-6\mathcal{P}\bar{\phi}_{A'B'}\Theta_{AB}-6\bar{\phi}_{A'B'}\mathcal{P} \Theta_{AB}+\tau^{C}{}_{(A'}\mathcal{P}^{2} H_{B')ABC} +\tau_{D(A'}\mathcal{P}\mathcal{D}^{DC}H_{B')ABC}. 
\end{align*}
The first and second terms on the right hand side are zero as a
consequence of conditions \eqref{Z=0} and
\eqref{Zdecomposecondition}. The last term can be also shown to be
zero by commuting the derivatives and using \eqref{Hdecompose1},
\eqref{Hdecompose2} and \eqref{gradHdecompose}. This leaves
\begin{equation}
0=\mathcal{P} S_{AA'BB'}=\tau^{C}{}_{(A'}\mathcal{P}^{2} H_{B')ABC}.\label{gradSdecomp}
\end{equation}
Eliminating the primed indices by multiplying by factors of $\tau_{AA'}$ gives
\[
\tau_{(C|}{}^{A'}\mathcal{P}^{2}H_{A'AB|D)}=0
\]
Thus, if this condition is satisfied on $\mathcal{S}$, then we have
that $\mathcal{P} S_{AA'BB'}|_{\mathcal{S}}=0$. In the following we
investigate further the consequences of this condition. As in a
Gaussian gauge $\mathcal{P}\tau_{AA'}=0$ it readily follows that, in
fact, one has  
\[
\mathcal{P}^{2}\left(\tau_{(C|}{}^{A'}H_{A'AB|D)}\right)=0.
\]
Splitting into irreducible parts, one obtains two necessary conditions:
\begin{subequations}
\begin{align}
&\mathcal{P}^{2}\xi_{ABCD}=0, \label{gradSdecompose1} \\
&\mathcal{P}^{2}\left(\mathcal{P}\kappa_{AB}+\frac{2}{3}\xi_{AB}\right)=0. \label{gradSdecompose2}
\end{align}
\end{subequations}

Let us first consider condition \eqref{gradSdecompose1}. We can
commute the Sen derivative with one of the normal derivatives to obtain
\begin{align*}
\mathcal{P}\big(\mathcal{P}\xi_{ABCD}\big)&=\mathcal{P}\big(\mathcal{P}\mathcal{D}_{(AB}\kappa_{CD)}\big) \\
&=\mathcal{P}\bigg(2\Psi_{(ABC}{}^{F}\kappa_{D)F}-2\Theta_{(AB}\widehat{\phi}_{CD)}-\frac{1}{3}\Omega_{ABCD}\xi-\frac{2}{3}\Omega_{F(ABC}\xi^{F}{}_{D)} \\
&\qquad-\frac{1}{3}\Omega_{(AB}\xi_{CD)}-\frac{1}{3}K\xi_{ABCD}+\Omega^{F}{}_{(A}\xi_{BCD)F} -\Omega_{(AB}{}^{EF}\xi_{CD)EF}+\mathcal{D}_{(AB}\mathcal{P}\kappa_{CD)}\bigg).
\end{align*}
\begin{align*}
\mathcal{P}\big(\mathcal{P}\xi_{ABCD}\big)&=\mathcal{P}\big(\mathcal{P}\mathcal{D}_{(AB}\kappa_{CD)}\big) \\
&=\mathcal{P}\bigg(2\Psi_{(ABC}{}^{F}\kappa_{D)F}-2\Theta_{(AB}\widehat{\phi}_{CD)}-\frac{1}{3}\Omega_{ABCD}\xi-\frac{2}{3}\Omega_{F(ABC}\xi^{F}{}_{D)} \\
&\qquad -\frac{1}{3}A_{(AB}\xi_{CD)}-\frac{1}{3}\Omega_{(AB}\xi_{CD)}-\frac{1}{3}K\xi_{ABCD}+A^{F}_{(A}\xi_{BCD)F}+\Omega^{F}{}_{(A}\xi_{BCD)F} \\
&\qquad -\Omega_{(AB}{}^{EF}\xi_{CD)EF}-\frac{1}{2}A_{(AB}\nabla\kappa_{CD)}+\mathcal{D}_{(AB}\nabla\kappa_{CD)}\bigg).
\end{align*}
Now, we can use our previous conditions on $\mathcal{S}$ to eliminate
terms. For example, the second term in the bracket is zero from
conditions \eqref{Z=0} and \eqref{Zdecomposecondition}. The fifth,
sixth and seventh terms vanish from \eqref{Hdecompose1} and
\eqref{gradHdecompose}. We can also use \eqref{Hdecompose2} and
\eqref{gradHdecompose2} to replace the last term ---alternatively, one
can commute the derivatives, use the substitution and then commute
back; the result is the same. From this substitution one obtains a
factor $\mathcal{D}_{(AB}\xi_{CD)}$ inside the normal derivative,
which can be replaced using \eqref{totalsym} ---this equation is valid
on the whole spacetime rather than just the hypersurface, so one is
allowed to take normal derivatives of it.

Proceeding as above, condition \eqref{gradSdecompose1} can be reduced to
\begin{equation}
\mathcal{P}^{2}\xi_{ABCD}=\mathcal{P}\left(4\Psi_{(ABC}{}^{F}\kappa_{D)F}\right)=0.
\label{gradSdecompose1Intermmediate}
\end{equation}
Now, splitting the covariant derivatives in the Bianchi identity
\eqref{BianchiIdentity} into normal and tangential components gives
the following space-spinor version:
\begin{equation*}
\mathcal{P}\Psi_{ABCD}=-4\widehat{\phi}_{F(A}\mathcal{D}^{F}{}_{B}\phi_{CD)}-4\widehat{\phi}_{(AB}\mathcal{D}^{F}{}_{C}\phi_{D)F}-2\mathcal{D}_{F(A}\Psi_{BCD)}{}^{F}.
\end{equation*}
One can use the latter expression to further reduce condition
\eqref{gradSdecompose1Intermmediate} to 
\begin{equation}
\Psi_{F(ABC}\xi_{D)}{}^{F}+6\widehat{\phi}_{F(A}\kappa^{E}{}_{B}\mathcal{D}^{F}{}_{C}\phi_{D)E} +6\widehat{\phi}_{(AB}\kappa^{E}{}_{C}\mathcal{D}^{F}{}_{D)}\phi_{EF}+3\kappa_{(A}{}^{F}\mathcal{D}_{B|E}\Psi_{F|CD)}{}^{E}=0. \label{gradSdecomposecondition}
\end{equation}
This is an intrinsic condition on $\mathcal{S}$.

\medskip
In order to obtain insight into condition \eqref{gradSdecompose2} we
make use, again, of the wave equation \eqref{WaveEquationKappa} for the spinor
$\kappa_{AB}$. Taking a normal derivative of this equation one obtains
\[
\mathcal{P}\left(\square\kappa_{AB}+\Psi_{ABCD}\kappa^{CD}\right)=0.
\]
Splitting the spacetime derivatives into normal and tangential parts and rearranging gives
\begin{align*}
\mathcal{P}\left(\mathcal{P}^{2}\kappa_{AB}\right)=\mathcal{P}&\big(-2 \kappa^{CD} \Psi_{ABCD} +\frac{2}{3} \Omega_{AB} \xi -  \frac{4}{3} \Omega_{(A}{}^{C} \xi_{B)C} + 2 \Omega^{CD} \xi_{ABCD}\\
&  -  K \mathcal{P}\kappa_{AB} - 
 \frac{2}{3} \mathcal{D}_{AB}\xi -  \frac{4}{3} 
\mathcal{D}_{C(A}\xi_{B)}{}^{C} - 2 \mathcal{D}_{CD}\xi_{AB}{}^{CD}\big).
\end{align*}
As before, we can use our previous conditions to eliminate terms. The
fourth and eight terms on the right hand side vanish due to
\eqref{Hdecompose1} and \eqref{gradHdecompose}. Also, we can use
equation \eqref{symcontract} to replace the the seventh term ---this
is because the relation \eqref{symcontract} holds on the whole
spacetime, and so one can take normal derivatives of it freely. These
steps give
\[
\mathcal{P}\big(\mathcal{P}^{2}\kappa_{AB}\big)=\mathcal{P}\bigg(\frac{2}{3}\Omega_{(A}{}^{C}\xi_{B)C}-\frac{2}{9}K\xi_{AB}
                                                                -\frac{2}{3}\Omega_{ABCD}\xi^{CD}+\frac{2}{3}\mathcal{D}_{AB}\xi\bigg).
\]
Alternatively, consider the second derivative of $\xi_{AB}$, given by
applying a normal derivative to equation \eqref{Normxi2} ---note that
equation \eqref{Normxi2} applies on the whole spacetime), so one can
take the normal derivative. This yields
\begin{align*}
\mathcal{P}^{2}\xi_{AB} &= \mathcal{P}\bigg(\frac{3}{2} \kappa^{CD} \Psi_{ABCD} -  3  \Theta_ {C(A}\widehat{\phi}_{B)}{}^{C}  -  \frac{1}{2} \Omega_{AB} \xi -  \frac{1}{3} K \xi_{AB} + \frac{1}{2} \Omega_{(A}{}^{C} \xi_{B)C} + \frac{1}{2} \Omega_{ABCD} \xi^{CD} \\
&\qquad + \frac{3}{4} \Omega^{CD} \xi_{ABCD} -  \frac{3}{2} \Omega_{(A}{}^{CDF} \xi_{B)CDF} -  \frac{3}{2} \mathcal{D}_{C(A}\mathcal{P}\kappa_{B)}{}^{C}\bigg).
\end{align*}
As before, we can use the conditions \eqref{Hdecompose1},
\eqref{Hdecompose2}, \eqref{gradHdecompose} and
\eqref{gradHdecompose2}, and the identity \eqref{symcontract} to reduce this to
\[
\mathcal{P}^{2}\xi_{AB}=\mathcal{P}\bigg(\frac{1}{3}K\xi_{AB}-\Omega_{(A}{}^{C}\xi_{B)C} + \Omega_{ABCD}\xi^{CD} - \mathcal{D}_{AB}\xi\bigg).
\]
By comparing terms, we find that 
\[
\mathcal{P}^{3}\kappa_{AB}=-\frac{2}{3}\mathcal{P}^{2}\xi_{AB}
\]
which is exactly the second condition \eqref{gradSdecompose2}. So, no further conditions are needed to be prescribed on the hypersurface ---this condition arises naturally from the evolution equation for the Killing spinor.

\subsubsection{Decomposing $\zeta_{AA'}=0$}
Recalling the definition of $\zeta_{AA'}$, equation \eqref{Definition:zeta}, and
splitting the spacetime spinorial  derivative into normal and
tangential parts one obtains 
\begin{align*}
\zeta_{AA'}&=\nabla^{B}{}_{A'}\Theta_{AB} \\
&=\frac{1}{2}\tau^{B}{}_{A'}\mathcal{P} \Theta_{AB}-\tau^{C}{}_{A'}\mathcal{D}_{C}{}^{B}\Theta_{AB}.
\end{align*}
From conditions \eqref{Z=0} and \eqref{Zdecomposecondition} it then follows that
$\zeta_{AA'}|_{\mathcal{S}}=0$. 

\subsubsection{Decomposing $\nabla_{EE'}\zeta_{AA'}=0$}
Again, if $\zeta_{AA'}|_{\mathcal{S}}=0$ then one one only needs to
consider the transverse derivative $\mathcal{P}\zeta_{AA'}$. By definition one has that 
\begin{align*}
\mathcal{P}\zeta_{AA'}&=\mathcal{P}\nabla^{B}{}_{A'}\Theta_{AB} \\
& = \mathcal{P}\bigg(-\tau^{C}{}_{A'}\mathcal{D}^{B}{}_{C}+\frac{1}{2}\tau^{B}{}_{A'}\mathcal{P}\bigg)\Theta_{AB} \\
& = \frac{1}{2}\tau^{B}{}_{A'}\mathcal{P}^{2}\Theta_{AB}
\end{align*}
where the last equation has been obtained by commuting the Sen and
normal derivatives, and using \eqref{Zdecomposecondition}. Therefore
one only needs to show that
\begin{equation*}
\mathcal{P}^{2}\Theta_{AB}=0.
\end{equation*}
Now, recalling the wave equation for $\Theta_{AB}$, equation \eqref{WaveEquation:Theta},
one readily notices that the right hand side vanishes on $\mathcal{S}$ as a consequence
of \eqref{Hdecompose1}, \eqref{Hdecompose2} and
\eqref{gradHdecompose}, so that
one is left with 
\[
\square \Theta_{AB}|_{\mathcal{S}} =0.
\]
Finally, expanding the left hand side one finds that on $\mathcal{S}$
\begin{align*}
\square \Theta_{AB} & = \nabla^{CC'}\nabla_{CC'}\Theta_{AB} \\
& = \left(-\tau^{BC'}\mathcal{D}^{C}{}_{B}+\frac{1}{2}\tau^{CC'}\mathcal{P}\right)\left(-\tau^{B}{}_{C'}\mathcal{D}_{BC}+\frac{1}{2}\tau_{CC'}\mathcal{P}\right)\Theta_{AB} \\
& = \frac{1}{4}\tau^{CC'}\tau_{CC'}\mathcal{P}^{2}\Theta_{AB}
\end{align*}
where the last line follows by commuting the derivatives where
appropriate and using conditions \eqref{Z=0} and
\eqref{Zdecomposecondition}. Finally, as $\tau^{CC'}\tau_{CC'}=2$ by
definition, we get that $\mathcal{P}^{2}\Theta_{AB}=0$ as a consequence of
the evolution equation for $\Theta_{AB}$.


\subsection{Eliminating  redundant conditions}

The discussion of the previous subsections can be summarised in the
following:

\begin{theorem}
\label{Theorem:IntrinsicKillingSpinorDataV1}
Let $\kappa_{AB}$ denote a Killing spinor candidate on an
electrovacuum spacetime $(\mathcal{M},\bmg,\bmF)$. If $\kappa_{AB}$
satisfies on a Cauchy hypersurface $\mathcal{S}$ the intrinsic conditions
\begin{subequations}
\begin{align}
\xi_{ABCD}&=0, \\
\Psi_{F(ABC}\kappa_{D)}{}^{F}&=0, \\
\kappa_{(A}{}^{C}\phi_{B)C}&=0, \\
\kappa_{(A|}{}^{C}\mathcal{D}_{CD}\phi_{|B)}{}^{D}&=\frac{1}{3}\phi_{(A}{}^{C}\xi_{B)C},\label{shortcondition} \\
3\kappa_{(A}{}^{F}\mathcal{D}_{B}{}^{E}\Psi_{CD)EF}+\Psi_{(ABC}{}^{F}\xi_{D)F}&=6\widehat{\phi}_{F(A}\kappa^{E}{}_{B}\mathcal{D}^{F}{}_{C}\phi_{D)E} +6\widehat{\phi}_{(AB}\kappa^{E}{}_{C}\mathcal{D}^{F}{}_{D)}\phi_{EF},\label{longcondition}
\end{align}
\end{subequations}
and its normal derivative at $\mathcal{S}$ is given by
\[
\mathcal{P}\kappa_{AB} = -\frac{2}{3}\xi_{AB},
\]
then $\kappa_{AB}$ is, in fact, a Killing spinor. 
\end{theorem}

\begin{remark}
{\em We note that
\[
\Theta_{AB}=\kappa_{(A}{}^{C}\phi_{B)C}=0\quad\mbox{implies}\quad
\phi_{AB}\propto\kappa_{AB}
\]
 Using this fact, one can show that
\eqref{shortcondition} and \eqref{longcondition} can be more simply
expressed as a condition on the proportionality between the Killing
spinor $\kappa_{AB}$ and the Maxwell spinor $\phi_{AB}$.}
\end{remark}

In order to simplify the conditions in Theorem
\ref{Theorem:IntrinsicKillingSpinorDataV1} and to analyse their
various interrelations we proceed by looking at the different
algebraic types that the Killing spinor can have.  First, we consider
the algebraically general case:

\begin{lemma}
\label{Lemma:Redundancy1}
Assume that a symmetric spinor $\kappa_{AB}$ 
satisfies the conditions 
\[
\kappa_{AB}\kappa^{AB}\neq0,\qquad \xi_{ABCD}=\Psi_{F(ABC}\kappa_{D)}{}^{F}=\kappa_{(A}{}^{C}\phi_{B)C}=0
\]
on an open subset $\mathcal{U}\subset\mathcal{S}$. Then, there exists a spin
basis $\{o^{A},\iota^{A}\}$ with $o_{A}\iota^{A}=1$ such that the
spinors $\kappa_{AB}$ and $\phi_{AB}$ can be expanded as 
\[
\kappa_{AB}=e^{\varkappa}o_{(A}\iota_{B)},\qquad\phi_{AB}=\varphi o_{(A}\iota_{B)}.
\]
Furthermore, if $\mathfrak{Q}\equiv\varphi e^{2\varkappa}$ is a constant
on $\mathcal{U}$, then conditions \eqref{shortcondition} and
\eqref{longcondition} are satisfied on $\mathcal{U}$.
\end{lemma}

\begin{proof}
The first part of the lemma follows directly from
$\kappa_{AB}\kappa^{AB}\neq0$, and the fact that
$\kappa_{(A}{}^{C}\phi_{B)C}=0$ implies that
$\phi_{AB}\propto\kappa_{AB}$. The condition
$\Psi_{F(ABC}\kappa_{D)}{}^{F}=0$ also allows us to expand the Weyl
spinor in the same basis:
\[
\Psi_{ABCD}=\psi o_{(A}o_{B}\iota_{C}\iota_{D)}.
\]

To show the redundancy of \eqref{shortcondition} and
\eqref{longcondition}, we first decompose the equation
$\mathcal{D}_{(AB}\kappa_{CD)}=0$ into irreducible components. To
simplify the notation, we borrow the $D, \Delta, \delta$ symbols from
the Newman-Penrose formalism to represent directional derivatives:
\begin{equation}
D\equiv o^{A}o^{B}\mathcal{D}_{AB},\qquad\Delta\equiv \iota^{A}\iota^{B}\mathcal{D}_{AB},\qquad \delta\equiv o^{A}\iota^{B}\mathcal{D}_{AB}.
\end{equation}
The components of $\mathcal{D}_{(AB}\kappa_{CD)}=0$ then become:
\begin{subequations}
\begin{align}
o^{C}Do_{C}&=0, \label{SKSEdecomp}\\
o^{C}\delta o_{C}&= -\frac{1}{2} D\varkappa, \\
\iota^{C}D\iota_{C}-o^{C}\Delta o_{C} &= 2\delta\varkappa, \\
\iota^{C}\delta\iota_{C}&=\frac{1}{2}\Delta\varkappa, \\
\iota^{C}\Delta\iota_{C} &=0.
\end{align}
\end{subequations}
Using these, one can show that
\[
e^{-\varkappa}\xi_{AB}=-3o_{A}o_{B}\iota^{F} \delta\iota_{F}-3\iota_{A}\iota_{B}o^{F} \delta o_{F} +\frac{3}{2}o_{(A}\iota_{B)}\left(\iota^{F}D\iota_{F} +o^{F}\Delta o_{F}\right).
\]
Now, using the electromagnetic Gauss constraint, equation
\eqref{MaxwellEquationDecomposition2}, together with the expansion for
$\phi_{AB}$ one obtains that using the basis expansion for $\phi_{AB}$
one obtains
\begin{equation}
\delta\varphi +2\varphi\delta\varkappa=0 \label{Maxwellconstraint}
\end{equation}
on $\mathcal{S}$. Now, the spacetime Bianchi identity
\eqref{BianchiIdentity} implies the contraint
\begin{equation}
\mathcal{D}^{CD}\Psi_{ABCD}=-2\hat{\phi}^{CD} \mathcal{D}_{CD}\phi_{AB} \label{Bianchispace}
\end{equation}
on $\mathcal{S}$. To find the basis expansion of the Hermitian
conjugate $\widehat{\phi}_{AB}$, note that:
\[
o_{A}\ho^{A}\equiv o_{A}\tau^{AA'}\bar{o}_{A'}=\tau_{AA'}o^{A}\bar{o}^{A'}=\tau_{a}k^{a}
\]
where $k_{a}\equiv o_{A}\bar{o}_{A'}$. As $\tau_{a}$ is timelike and
$k_{a}$ is null, this scalar product is non-zero, and the pair
$\{o_{A},\ho_{A}\}$ forms  a basis. We expand the spinor $\iota^{A}$
in this basis as 
\[
\iota^{A}=\alpha\ho^{A}+\beta o^{A}
\]
Contracting this with $o_{A}$, we find $1/\alpha=o_{A}\ho^{A}\geq0$,
and so $\alpha\geq0$. Performing a Lorentz transformation on the basis
$\{o_{A},\iota_{A}\}$ parametrised by the complex function one has
that 
\begin{align*}
o^{A}\mapsto \tilde{o}&=\frac{1}{\lambda} o^{A}, \\
\iota^{A}\mapsto \tilde{\iota}&=\lambda\iota^{A}.
\end{align*}
This transformation preserves the value of $o_{A}\iota^{A}$ and the
symmetrised product $o_{(A}\iota_{B)}$, and thus, it preserves the
form of the basis expansions of $\kappa_{AB}$ and
$\phi_{AB}$. Moreover, one has that 
\[
\tilde{\iota}^{A}=\alpha|\lambda|^2\widehat{\tilde{o}}^{A}+\beta\lambda^2 \tilde{o}^{A}.
\]
So, by choosing $|\lambda|^2=1/\alpha$ and $\tilde{\beta}=\beta\lambda^2$, and dropping the tildes, we get
\begin{align*}
\iota^{A}&=\widehat{o}^{A}+\beta o^{A}, \\
\widehat{\iota}^{A}&=-o^{A}+\bar{\beta}\ho^{A}.
\end{align*}
Using the above expressions we can find the basis expansion of
$\hat{\phi}_{AB}$. Namely, one has that:
\begin{align*}
\widehat{\phi}_{AB}&=\frac{1}{2}\bar{\varphi}(\ho_{A}\widehat{\iota}_{B}+\widehat{\iota}_{A}\ho_{B}) \\
&=\frac{1}{2}\bar{\varphi}(-o_{A}\ho_{B}-\ho_{A}o_{B}+2\bar{\beta}\ho_{A}\ho_{B}) \\
&=\bar{\varphi}\bar{\beta}\iota_{A}\iota_{B}+\bar{\varphi}\beta(1+|\beta|^2)o_{A}o_{B}-\bar{\varphi}(1+2|\beta|^2)o_{(A}\iota_{B)}.
\end{align*}

Now, using the basis expansion for the Weyl spinor, contracting with
combinations of $o^{A}$ and $\iota^{A}$ and using the relations given
in \eqref{SKSEdecomp} and \eqref{Maxwellconstraint}, the components of
\eqref{Bianchispace} become 
\begin{align*}
D\psi +3\psi D\varkappa &= 6|\varphi|^2(1+2|\beta|^2) D\varkappa +12\bar{\beta}|\varphi|^2o^{A}\Delta o_{A}, \\
\Delta\psi +3\psi\Delta\varkappa &= 6|\varphi|^2(1+2|\beta|^2) \Delta\varkappa -12\beta|\varphi|^2(1+|\beta|^2)\iota^{A}D\iota_{A}, \\
\delta\psi +3\psi \delta\varkappa &=-6|\varphi|^2(1+2|\beta|^2)\delta\varkappa -3\beta\bar{\varphi}(1+|\beta|^2)D\varphi-3\bar{\beta}\bar{\varphi}\Delta\varphi.
\end{align*}
Exploiting the conditions \eqref{SKSEdecomp}, the expansions of the
Maxwell and the Bianchi constraints it can be shown that condition \eqref{longcondition} can be decomposed into the following non-trivial irreducible parts:
\begin{align*}
\bar{\beta}\bar{\varphi}\left(D\varphi+2\varphi D\varkappa\right)&=0, \\
\bar{\varphi}(1+2|\beta|^2)\left(D\varphi+2\varphi D\varkappa\right)&=0, \\
\bar{\varphi}(1+2|\beta|^2)\left(\Delta\varphi +2\varphi \Delta\varkappa\right)&=0, \\
\beta\bar{\varphi}(1+|\beta|^2)\left(\Delta\varphi+2\varphi\Delta\varkappa\right)&=0. 
\end{align*}
Assuming $\varphi\neq0$, these conditions along with the Maxwell
constraint \eqref{Maxwellconstraint} are equivalent to the
following basis-independent expression, also independent of the value
of $\beta$:
\[
\mathcal{D}_{AB}\varphi+2\varphi\mathcal{D}_{AB}\varkappa=0.
\]
The latter can be written as
\[ 
\mathcal{D}_{AB}\left(\varphi e^{2\varkappa}\right)=\mathcal{D}_{AB}\mathfrak{Q}=0.
\] 
Therefore, under the hypotheses of the present lemma, equation
\eqref{longcondition} is equivalent to the requirement of
$\mathfrak{Q}$ being constant in a domain
$\mathcal{U}\subset\mathcal{S}$. In a similar way, substituting the
above relations in equation \eqref{shortcondition} and splitting into
irreducible parts gives the following set of equivalent conditions:
\begin{align*}
e^{\varkappa}\left(D\varphi+2\phi D\varkappa\right) &=0, \\
e^{\varkappa}\left(\Delta\varphi+2\varphi \Delta\varkappa\right)  &=0, \\
e^{\varkappa}\left(\delta\varphi+2\varphi \delta\varkappa\right) &=0.
\end{align*}
As $e^{\varkappa}$ is non-zero, this set of conditions is again
equivalent to the constancy of $\mathfrak{Q}$ in
$\mathcal{U}\subset\mathcal{S}$.
\end{proof}

Next, we consider the case when the Killing spinor is algebraically special:

\begin{lemma}
\label{Lemma:Redundancy2}
Assume the symmetric spinor $\kappa_{AB}$ satisfies the conditions
\[
\kappa_{AB}\kappa^{AB}=0,\qquad\kappa_{AB}\hat{\kappa}^{AB}\neq0,\qquad\xi_{ABCD}=\Psi_{F(ABC}\kappa_{D)}{}^{F}=\kappa_{(A}{}^{C}\phi_{B)C}=0
\]
on an open subset $\mathcal{U}\subset\mathcal{S}$. Then, there exists
a normalised spin basis $\{o^{A},\iota^{A}\}$  such that the spinors
$\kappa_{AB}$ and $\phi_{AB}$ can be expanded as
\[
\kappa_{AB}=e^{\varkappa}o_{A}o_{B},\qquad\phi_{AB}=\varphi o_{A}o_{B}.
\]
Furthermore, the equations \eqref{shortcondition} and
\eqref{longcondition} are satisfied on
$\mathcal{U}\subset\mathcal{S}$.
\end{lemma}

\begin{proof}
The first part of the lemma follows directly from the hypothesis
$\kappa_{AB}\kappa^{AB}=0,\, \kappa_{AB}\widehat{\kappa}^{AB}\neq0$,
and the fact that $\kappa_{(A}{}^{C}\phi_{B)C}=0$ implies
$\phi_{AB} \propto\kappa_{AB}$. The condition
$\Psi_{F(ABC}\kappa_{D)}{}^{F}=0$ also allows us to expand the Weyl
spinor in the same basis as
\[
\Psi_{ABCD}=\psi o_{A}o_{B}o_{C}o_{D}.
\]
In this basis, the components of the equation $\mathcal{D}_{(AB}\kappa_{CD)}=0$ become
\begin{align*}
&o^{A}Do_{A}=0, \\
&D\varkappa+4o^{A}\delta o_{A}+2\iota^{A}Do_{A}=0, \\
&\delta\varkappa+o^{A}\Delta o_{A}+2\iota^{A}\delta o_{A}=0, \\
&\Delta\varkappa+2\iota^{A}\Delta o_{A}=0.
\end{align*}
Using these relations one can show that
\[
e^{-\varkappa}\xi_{AB}=3o_{A}o_{B}o^{C}\Delta o_{C}
-6o_{(A}\iota_{B)}o^{C}\delta o_{C}. 
\]
The Maxwell constraint, equation
\eqref{MaxwellEquationDecomposition2}, on $\mathcal{S}$ is equivalent
to
\[
D\phi-\phi D\varkappa-6\phi o^{A}\delta o_{A}=0,
\]
and the $o_{(A}\iota_{B)}$ component of the Bianchi constraint 
\[
\mathcal{D}^{CD}\Psi_{ABCD}=-2\widehat{\phi}^{CD}
\mathcal{D}_{CD}\phi_{AB}
\]
 on $\mathcal{S}$, as a consequence of the previous relations, is equivalent to the following condition:
\begin{equation*}
|\varphi|^2 o^{A}\Delta o_{A}-2\beta|\varphi|^2o^{A}\delta o_{A}=0.
\end{equation*}
Then, by substituting all the relevant basis expansions into
\eqref{shortcondition} and \eqref{longcondition}, and splitting the
equations into irreducible parts, one finds that both conditions are
automatically satisfied as a result of the above relations.
\end{proof}

We round up the discussion of this section with the following
electrovacuum analogue of Theorem in \cite{BaeVal12}:

\begin{lemma}
Assume that one has a symmetric spinor $\kappa_{AB}$ satisfying the conditions
\[
\mathcal{D}_{(AB}\kappa_{CD)}=\Psi_{F(ABC}\kappa_{D)}{}^{F}=\kappa_{(A}{}^{C}\phi_{B)C}=0
\]
on the Cauchy
hypersurface $\mathcal{S}$ and that the complex
function 
\[
\mathfrak{Q}^{2}\equiv \left(\kappa_{AB}\kappa^{AB} \right)^{2}
\phi_{AB}\phi^{AB}
\]
 is constant on $\mathcal{S}$. Then the domain of dependence, $D^+(\mathcal{S})$, of the
initial data set $(\mathcal{S},\bmg,\bmK,\bmF)$ will admit a Killing
spinor.
\end{lemma}

\begin{proof}
Let $\mathcal{U}_1$ be the set of all points in $\mathcal{S}$ where
$\kappa_{AB}\kappa^{AB}\neq0$ and $\mathcal{U}_2$ be the set of all
points in $\mathcal{S}$ where
$\kappa_{AB}\widehat{\kappa}^{AB}\neq0$. The scalar functions
$\kappa_{AB}\kappa^{AB}:\mathcal{S}\rightarrow\mathbb{C}$ and
$\kappa_{AB}\widehat{\kappa}^{AB}:\mathcal{S}\rightarrow\mathbb{R}$
are continuous. Therefore, $\mathcal{U}_1$ and $\mathcal{U}_2$ are
open sets. Now, let $\mathcal{V}_1$ and $\mathcal{V}_2$ denote,
respectively, the interiors of $\mathcal{S} \setminus \mathcal{U}_1$
and $\mathcal{V}_1 \setminus \mathcal{U}_2$. On the open set
$\mathcal{V}_1 \cap \mathcal{U}_2$, we have that
$\kappa_{AB}\kappa^{AB}=0$ and
$\kappa_{AB}\widehat{\kappa}^{AB}\neq0$. Hence, by Lemma \ref{Lemma:Redundancy2}, the
conditions \eqref{shortcondition} and \eqref{longcondition} are
satisfied on $\mathcal{V}_1 \cap \mathcal{U}_2$. Similarly, by Lemma
\ref{Lemma:Redundancy1}, conditions \eqref{shortcondition} and \eqref{longcondition} are
satisfied on $\mathcal{U}_1$. On the open set $\mathcal{V}_2$, we have
that $\kappa_{AB}=0$ and therefore \eqref{shortcondition} and
\eqref{longcondition} are trivially satisfied on
$\mathcal{V}_2$. Using the above sets, the 3-manifold $\mathcal{S}$
can be split as
\[
\text{int}\mathcal{S}=\mathcal{U}_{1}\cup(\mathcal{V}_{1}\cap
\mathcal{U}_{2})\cup \mathcal{V}_{2}\cup\partial
\mathcal{U}_{1}\cup\partial \mathcal{V}_{2}. 
\]
By hypothesis, all terms in conditions \eqref{shortcondition} and
\eqref{longcondition} are continuous, and the conditions themselves
are satisfied on the open sets $\mathcal{U}_{1}$, $\mathcal{V}_{2}$
and $\mathcal{V}_{1}\cap \mathcal{U}_{2}$. By continuity, the
conditions are also satisfied on the boundaries $\partial
\mathcal{U}_{1}$ and $\partial \mathcal{V}_{2}$. Therefore,
\eqref{shortcondition} and \eqref{longcondition} are satisfied on
$\text{int}\,\mathcal{S}$, and by continuity this extends to the
whole of $\mathcal{S}$.
\end{proof}

\subsection{Summary}
We can summarise the discussion of the present section calculations in
the following theorem:

\begin{theorem}
\label{Theorem:KillingSpinorData}
Let $(\mathcal{S}, \bmh, \bmK,\bmF)$ be an initial data set for the
Einstein-Maxwell field equations where $\mathcal{S}$ is a Cauchy
hypersurface. If the conditions
\begin{subequations}
\begin{align}
\xi_{ABCD}&=0, \label{ReducedKillingSpinorDataCondition1}\\
\Psi_{F(ABC}\kappa_{D)}{}^{F}&=0, \label{ReducedKillingSpinorDataCondition2}\\
\kappa_{(A}{}^{C}\phi_{B)C}&=0, \label{ReducedKillingSpinorDataCondition3}\\
\mathfrak{Q}^{2}\equiv\left(\kappa_{AB}\kappa^{AB} \right)^{2}
  \phi_{AB}\phi^{AB}&=\text{constant}, \label{ReducedKillingSpinorDataCondition4}
\end{align}
\end{subequations}
are satisfied on $\mathcal{S}$, then the development of the initial
data set will admit a Killing spinor in the domain of dependence of
$\mathcal{S}$. The Killing spinor is obtained by evolving
\eqref{WaveEquationKappa} with initial data satisfying the
above conditions and
\[
\mathcal{P}\kappa_{AB}=-\frac{2}{3}\xi_{AB}
\]
on $\mathcal{S}$.
\end{theorem}

\section{The approximate Killing spinor equation}
\label{Section:ApproximateKillingEquation}

In the previous section we have identified the conditions that need to
be satisfied by an initial data set for the Einstein-Maxwell equations
so that its development is endowed with a Killing spinor ---see
Theorem \ref{Theorem:KillingSpinorData}. Together with the
characterisation of the Kerr-Newman spacetime given by Theorem
\ref{Theorem:KerrNewmanCharacterisation}, the latter provide a way of
characterising initial data for the Kerr-Newman spacetime. The key
equation in this characterisation is the \emph{spatial Killing spinor
equation}
\[
\mathcal{D}_{(AB} \kappa_{CD)}=0.
\]
As it will be seen in the following, this equation is overdetermined
and thus, admits no solution for a generic initial data set
$(\mathcal{S},\bmh,\bmK,\bmF)$. Following the discussion of Section 5
in \cite{BaeVal10b}, in this section we show how to construct a
  elliptic equation for a spinor $\kappa_{AB}$ over $\mathcal{S}$
  which can always be solved and which provides, in some sense, a best
  fit to a spatial Killing spinor. This \emph{approximate Killing
    spinor} will be used, in turn, to measure the deviation of the
  electrovacuum initial data set under consideration from initial data
  for the Kerr-Newman spacetime.

\subsection{Basic identities}
In the present section we provide a brief discussion of the
basic ellipticity properties of the spatial Killing equation. In what
follows, let $\mathfrak{S}_{(AB)}(\mathcal{S})$ and
$\mathfrak{S}_{(ABCD)}(\mathcal{S})$ denote, respectively, the space
of totally symmetric valence 2 and 4 spinor fields over the 3-manifold
$\mathcal{S}$. Given $\mu_{AB},\, \nu_{AB} \in
\mathfrak{S}_{(AB)}(\mathcal{S})$, $\zeta_{ABCD}, \, \chi_{ABCD}\in
\mathfrak{S}_{(ABCD)}(\mathcal{S})$ one can use the Hermitian
structure induced on $\mathcal{S}$ by $\tau^{AA'}$ to define an inner
product in $\mathfrak{S}_{(AB)}(\mathcal{S})$ and
$\mathfrak{S}_{(ABCD)}(\mathcal{S})$, respectively, via
\begin{equation}
\langle \bmmu, \bmnu \rangle \equiv \int_{\mathcal{S}}
\mu_{AB} \widehat{\nu}^{AB} \mbox{d}\mu, \qquad \langle \bmzeta, \bmchi \rangle \equiv \int_{\mathcal{S}}
\zeta_{ABCD} \widehat{\chi}^{ABCD} \mbox{d}\mu
\label{InnerProductSpinorFields}
\end{equation}
where $\mbox{d}\mu$ denotes volume form of the 3-metric $\bmh$. 

\medskip
Let now $\Phi$ denote the \emph{spatial Killing spinor operator} 
\[
\Phi: \mathfrak{S}_{(AB)}(\mathcal{S}) \longrightarrow
\mathfrak{S}_{(ABCD)}(\mathcal{S}), \qquad \Phi(\bmkappa) \equiv
\mathcal{D}_{(AB} \kappa_{CD)}. 
\]
The inner product \eqref{InnerProductSpinorFields} allows to define
$\Phi^*: \mathfrak{S}_{(ABCD)}(\mathcal{S}) \longrightarrow
\mathfrak{S}_{(AB)}(\mathcal{S})$, the \emph{formal adjoint of $\Phi$},
through the condition
\[
\langle \Phi(\bmkappa), \bmzeta \rangle = \langle
\bmkappa, \Phi^*(\bmzeta)\rangle. 
\]
In order to evaluate the above condition one makes use of the identity
(\emph{integration by parts})
\begin{equation}
 \int_{\mathcal{U}} \mathcal{D}^{AB} \kappa^{CD}
\widehat{\zeta}_{ABCD} \mbox{d}\mu - \int_{\mathcal{U}} \kappa^{AB}
\widehat{\mathcal{D}^{CD}\zeta_{ABCD}}\mbox{d}\mu + \int_{\mathcal{U}}
2 \kappa^{AB} \Omega^{CDF}{}_A \widehat{\zeta}_{BCDF}\mbox{d}\mu =
\int_{\partial \mathcal{U}} n^{AB} \kappa^{CD} \widehat{\zeta}_{ABCD} \mbox{d}S
\label{IntegrationByParts}
\end{equation}
with $\mathcal{U}\subset\mathcal{S}$ and where $\mbox{d}S$ denotes the
area element of $\partial\mathcal{U}$, $n_{AB}$ is the spinorial
counterpart of its outward pointing normal and $\zeta_{ABCD}$ is a totally
symmetric spinorial field. Now, observing that 
\begin{eqnarray*}
&& \langle \Phi(\bmkappa), \bmzeta \rangle = \int_{\mathcal{S}}
\mathcal{D}_{(AB} \kappa_{CD)} \widehat{\zeta}^{ABCD} \mbox{d}\mu \\
&& \phantom{ \langle \Phi(\bmkappa), \bmzeta \rangle} =
\int_{\mathcal{S}} \mathcal{D}_{AB} \kappa_{CD} \widehat{\zeta}^{ABCD} \mbox{d}\mu, 
\end{eqnarray*}
it follows then from the identity \eqref{IntegrationByParts} that 
\[
\Phi^*(\bmzeta) = \mathcal{D}^{AB}\zeta_{ABCD} -2 \Omega^{ABF}{}_{(C} \zeta_{D)ABF}.
\]

\begin{definition}
The composition operator $L\equiv \Phi^* \circ \Phi:
\mathfrak{S}_{(AB)}(\mathcal{S}) \longrightarrow
\mathfrak{S}_{(AB)}(\mathcal{S})$ given by
\begin{equation}
L(\bmkappa) \equiv \mathcal{D}^{AB} \mathcal{D}_{(AB} \kappa_{CD)} -
\Omega^{ABF}{}_{(A} \mathcal{D}_{|DF|} \kappa_{B)C} -
\Omega^{ABF}{}_{(A}\mathcal{D}_{B)F} \kappa_{CD}
\label{ApproximateKillingSpinorEquation}
\end{equation}
will be called the \emph{approximate Killing spinor operator} and the
equation
\[
L(\bmkappa) =0
\]
the \emph{approximate Killing spinor equation}.
\end{definition}

\begin{remark}
{\em A direct computation shows that the approximate Killing spinor
  equation \eqref{ApproximateKillingSpinorEquation} is, in fact, the
  Euler-Lagrange equation of the functional}
\[
J\equiv \int_{\mathcal{S}} \mathcal{D}_{(AB}\kappa_{CD)} \widehat{\mathcal{D}^{AB}\kappa^{CD}}\mbox{d}\mu.
\]
\end{remark}

\subsection{Ellipticity of the approximate Killing spinor equation}

The key observation concerning the approximate Killing spinor operator
is given in the following:

\begin{lemma}
The operator $L$ is a formally self-adjoint elliptic operator.
\end{lemma}
\begin{proof}
It is sufficient to look at the principal part of the operator $L$
given by
\[
P(L)(\bmkappa)=\mathcal{D}^{AB}\mathcal{D}_{(AB}\kappa_{CD)}.
\]
The symbol for this operator is given by
\[
\sigma_{L}(\bmxi)\equiv\xi^{AB}\xi_{(AB}\kappa_{CD)}
\]
where the argument $\xi_{AB}$ satisfies $\xi_{AB}=\xi_{(AB)}$ and
$\widehat{\xi}_{AB}=-\xi_{AB}$---i.e. $\xi$ is a real symmetric
spinor. Also, define the inner product $\langle\,,\rangle$ on the space of symmetric valence-2 spinors by
\[
\langle\bmxi,\bmeta\rangle\equiv\widehat{\xi}^{AB}\eta_{AB}.
\]

The operator $L$ is elliptic if the map
\[
\sigma_{L}(\bmxi): \kappa_{AB}\mapsto\xi^{CD}\xi_{(CD} \kappa_{AB)}
\]
is an isomorphism when
$|\bmxi|^2\equiv\langle\bmxi,\bmxi\rangle\neq0$. As the above mapping
is linear and between vector spaces of the same dimension, one only
needs to verify injectivity ---in other words, that if
$\xi^{AB}\xi_{(AB}\kappa_{CD)}=0$, then $\kappa_{AB}=0$. To show this,
first expand the symmetrisation in the symbol to obtain
\[
-\kappa_{CD}|\bmxi|^2-\langle\bmxi,\bmkappa\rangle\xi_{CD}+2\xi^{AB} \xi_{CB}\kappa_{AD} +2\xi^{AB}\xi_{DB}\kappa_{AC}=0,
\]
where we have used the reality condition
$\widehat{\xi}_{AB}=-\xi_{AB}$. Note also that the Jacobi identity
implies that
\[
\xi^{AB}\xi_{CB}=-\frac{1}{2}\delta_{C}{}^{A}|\bmxi|^2
\]
which reduces the above equation to
\[
3\kappa_{CD}|\bmxi|^2+\xi_{CD}\langle\bmxi,\bmkappa\rangle=0.
\]
Contracting this with $\widehat{\kappa}^{CD}$, and using the conjugate
symmetry of the inner product, we obtain
\begin{equation*}
3|\bmkappa|^2|\bmxi|^2+|\langle\bmxi,\bmkappa\rangle|^2=0.
\end{equation*}
Both of these terms are positive, and so the equality can only hold if
each term vanishes individually. Taking the first of these, one sees
that when $|\bmxi|^2\neq0$, we must have $|\bmkappa|^2=0$. This is
equivalent to $\kappa_{AB}=0$, completing the proof of
injectivity and establishing the ellipticity of $L$.
\end{proof}

\section{The approximate Killing spinor equation in
  asymptotically Euclidean manifolds}
\label{Section:AsymptoticallyEuclideanManifolds}

The purpose of this section is to discuss the solvability of the
approximate Killing spinor equation, equation
\eqref{ApproximateKillingSpinorEquation}, in asymptotically Euclidean
manifolds. As a result of this analysis one concludes that for this
type of initial data sets for the Einstein-Maxwell equations it is
always possible to construct an approximate Killing spinor.

\subsection{Weighted Sobolev norms}
The discussion of asymptotic boundary conditions for the approximate
Killing equation makes use of \emph{weighted Sobolev norms} and
\emph{spaces}. In this section we introduce the necessary terminology
and conventions to follow the discussion. The required properties of these objects for the
present analysis are discussed in detail in Section 6.2  of
\cite{BaeVal10b} to which the reader is directed for further reference.

Given $u$ a scalar function over $\mathcal{S}$ and
$\delta\in \mathbb{R}$, let
$\parallel u \parallel_\delta$ denote the weighted $L^2$ Sobolev norm
of $u$. All throughout we make use of of Bartnik's conventions for
the weights ---see \cite{Bar86}--- so that, in particular $\parallel
u \parallel_{-3/2}$ is the \emph{standard} $L^2$ norm of
$u$. Similarly, let $H_\delta^s$ with $s$ a non-negative index denote
the weighted Sobolev space of functions for which the norm
\[
\parallel u \parallel_{s,\delta}\equiv \sum_{0\leq |\alpha |\leq
  s} \parallel D^\alpha u \parallel_{\delta-|\alpha|} 
\]
is finite where $\alpha=(\alpha_1,\alpha_2,\alpha_3)$ is a multiindex
and $|\alpha|\equiv \alpha_1+\alpha_2+\alpha_3$. We say that $u\in
H_\delta^\infty$ if $u\in H^s_\delta$ for all $s$. We will say that a
spinor or a tensor belongs to a function space if its norm does ---so,
for instance $\zeta_{AB}\in H^s_\delta$ is a shorthand for
$(\zeta_{AB}\hat{\zeta}^{AB} + \zeta_A{}^A \hat{\zeta}_B{}^B)^{1/2}\in
H^s_\delta$. A property of the weighted Sobolev spaces that will be
used repeatedly is the following: if $u\in H^\infty_\delta$ then $u$
is smooth (i.e. $C^\infty$ over $\mathcal{S}$) and has a fall off at
infinity such that $D^\alpha u = o (r^{\delta
  -|\alpha|})$\footnote{Recall that the \emph{small o} indicates that
  if $f(x)=o(x^n)$, then $f(x)/x^n\rightarrow 0$ as $x\rightarrow
  0$. }. In a slight abuse of notation, if $u\in H^\infty_\delta$ then
we will often say that $u=o_\infty (r^\delta)$ at a given asymptotic end. 

\subsection{Asymptotically Euclidean manifolds}

We begin by spelling out our assumptions on the class of
Einstein-Maxwell initial data sets to be considered in the remainder of
this article. The
Einstein-Maxwell constraint equations are given by
\begin{eqnarray*}
&& r - K^2 + K_{ij} K^{ij} = 2 \rho, \\
&& D^j K_{ij} -D_i K = j_i, \\
&& D^i E_i =0, \\
&& D^i B_i =0,
\end{eqnarray*}
where $D_i$ denotes the Levi-Civita connection of the 3-metric $\bmh$,
$r$ is the associated Ricci scalar, $K_{ij}$ is the extrinsic
curvature, $K\equiv K_i{}^i$, $\rho$ is the energy-density of the
electromagnetic field, $j_i$ is the associated Poynting vector and
$E_i$ and $B_i$ denote the electric and magnetic parts of the Faraday
tensor with respect to the unit normal of $\mathcal{S}$. 

\begin{assumption}
\label{Assumption:AsymptoticConditions}
In the remainder of this article it will be assumed that one has
initial data $(\bmh, \bmK,\bmE,\bmB)$ for the Einstein-Maxwell equations which is
asymptotically Reissner-Nordstr\"om in the sense that in each asymptotic end
of $\mathcal{S}$ there exist \emph{asymptotically Cartesian coordinates} $(x^\alpha)$ and two constants $m, q$ for which 
\begin{subequations}
\begin{eqnarray}
&& h_{\alpha\beta} = -\left( 1 +
  \frac{2m}{r}\right)\delta_{\alpha\beta} + o_\infty(r^{-3/2}), \label{AsymptoticEuclidean1}\\
&& K_{\alpha\beta} = o_\infty (r^{-5/2}), \label{AsymptoticEuclidea2}\\
&& E_\alpha = \frac{qx_{\alpha}}{r^2} + o_\infty (r^{-5/2}), \label{AsymptoticEuclidean3}\\
&& B_\alpha = o_\infty(r^{-5/2}). \label{AsymptoticEuclidean4}
\end{eqnarray}
\end{subequations}
\end{assumption}

\begin{remark}
{\em The asymptotic conditions spelled in Assumption
  \ref{Assumption:AsymptoticConditions} ensure that the total electric
  charge of the initial data is non-vanishing. In particular, it contains standard
  initial data for the Kerr-Newman spacetime in, say, Boyer-Lindquist
  coordinates as an example. More generally, the assumptions are
  consistent with the notion of {\em stationary asymptotically flat end}
  provided in Definition \ref{Definition:StationaryEnd}. }
\end{remark}

\begin{remark}
{\em The above class of initial data is not the most general one could
consider. In particular, conditions
\eqref{AsymptoticEuclidean1}-\eqref{AsymptoticEuclidean4} exclude
boosted initial data. In order to do so one would require that
\[
K_{\alpha\beta} = o_\infty(r^{-3/2}).
\] 
The Einstein-Maxwell constraint equations would then require one to modify
the leading behaviour of the the 3-metric $h_{\alpha\beta}$. The
required modifications for this extension of the present analysis are
discussed in \cite{BaeVal10b}. }
\end{remark}

\subsection{Asymptotic behaviour of the approximate Killing spinor}

In this section we discuss the asymptotic behaviour of solutions to
the spatial Killing spinor equation on asymptotically Euclidean
manifolds of the type described in Assumption
\ref{Assumption:AsymptoticConditions}. To this end, we first consider
the behaviour of the Killing spinor in the Kerr-Newman spacetime. In a
second stage we impose the same asymptotics solutions to the
approximate Killing spinor equation on slices of a more general
spacetime. In what follows, we concentrate our discussion to an
asymptotic end.

\subsubsection{Asymptotic behaviour in the exact Kerr-Newman spacetime}
 For the exact Kerr-Newman spacetime with
mass $m$, angular momentum $a$ and charge $q$ it is possible to
introduce a NP frame $\{l^a,\, n^a,\, m^a,\, \bar{m}^a \}$ with associated spin dyad
$\{ o^A,\, \iota^A\}$  such that the spinors $\kappa_{AB}$, $\phi_{AB}$
and $\Psi_{ABCD}$ admit the expansion
\[
\kappa_{AB} =\varkappa o_{(A}\iota_{B)}, \qquad \phi_{AB} = \varphi
o_{(A}\iota_{B)}, \qquad \Psi_{ABCD} = \psi o_{(A}o_B \iota_C \iota_{D)},
\]
with 
\begin{eqnarray*}
&& \varkappa = \frac{2}{3} (r - \mbox{i} a \cos\theta), \\
&&  \varphi = \frac{q}{(r-\mbox{i}a \cos \theta)^2}, \\
&& \psi = \frac{6}{(r -\mbox{i}a \cos\theta)^3}\left(  \frac{q^2}{r+
    \mbox{i} a \cos\theta } -m \right),
\end{eqnarray*}
where $r$ denotes the standard \emph{Boyer-Lindquist} radial
coordinate ---see \cite{AndBaeBlu15} for more details. A further
computation shows that the spinorial counterpart, $\xi^{AA'}$, of the
Killing vector $\xi^a$ takes the form
\begin{equation}
\xi_{AA'} =-\frac{3}{2}\varkappa ( \mu o_A \bar{o}_{A'} - \pi o_A
\bar{\iota}_{A'} + \tau \iota_A \bar{o}_{A'} - \rho \iota_A
\bar{\iota}_{A'})
\label{XiExpansion}
\end{equation}
where the NP spin connection coefficients $\mu$, $\pi$, $\tau$ and
$\rho$ satisfy the conditions
\[
\bar{\mu}\bar{\varkappa} = \mu\varkappa, \qquad
\bar{\tau}\bar{\varkappa}= \varkappa\pi, \qquad \bar{\rho}
\bar{\varkappa} = \varkappa \rho
\]
which ensure that $\xi_{AA'}$ is a Hermitian spinor
---i.e. $\xi_{AA'}=\bar{\xi}_{AA'}$. Despite the
conciseness of the above expressions, the basis of \emph{principal spinors}
given by  $\{o^A,\iota^A \}$ is not well adapted to the discussion of
asymptotics on a stationary end of the Kerr-Newman spacetime. 

\medskip
From the point of view of asymptotics, a better representation of the
Kerr-Newman spacetime is obtained using a NP frame $\{ l^{\prime
  a},\,n^{\prime a},\, m^{\prime a}, \, \bar{m}^{\prime a}  \}$ with associated spin dyad $\{ o^{\prime
  A},\, \iota^{\prime A} \}$ such that  
\[
\tau^a = l^{\prime a} + n^{\prime a} = \sqrt{2}(\bmpartial_t)^a,
\]
where the vector $\tau^a$ is the tensorial counterpart of the spinor
$\tau^{AA'}$. It follows from the above that
\begin{equation}
\tau^{AA'} =o^{\prime A}\bar{o}^{\prime A'} + \iota^{\prime A}
\bar{\iota}^{\prime A'}.
\label{TauExpansion}
\end{equation}
Notice, in particular, that from the above expression it follows that
$\iota'_A=\widehat{o'}_A$. As $\tau_{AA'} =\sqrt{2}\xi_{AA'}$, one can use the expressions
\eqref{XiExpansion} and \eqref{TauExpansion} to compute the leading
terms of the Lorentz transformation relating the NP frames $\{l^a,\,
n^a,\, m^a,\, \bar{m}^a \}$ and $\{ l^{\prime
  a},\,n^{\prime a},\, m^{\prime a}, \, \bar{m}^{\prime a}  \}$. The
details of this tedious computation will not be presented here ---just
the main result. 

In what follows it will be convenient to denote the spinors of the
basis $\{ o^{\prime
  A},\, \iota^{\prime A} \}$ in the form $\{  \epsilon_\bmA{}^A \}$
where
\[
\epsilon_\bmzero{}^A = o^{\prime A}, \qquad \epsilon_\bmone{}^A
=\iota^{\prime A}. 
\]
Moreover, let $\kappa_{\bmA\bmB}
\equiv \epsilon_\bmA{}^A\epsilon_\bmB{}^B\kappa_{AB}$ denote the components
of $\kappa_{AB}$ with respect to the basis $\{ \epsilon_\bmA{}^A
\}$. It can then be shown that for Kerr-Newman initial data satisfying
the asymptotic conditions
\eqref{AsymptoticEuclidean1}-\eqref{AsymptoticEuclidean4} one can
choose asymptotically Cartesian coordinates $(x^\alpha)=(x^1,x^2,x^3)$
and orthonormal frames on the asymptotic ends such that
\begin{equation}
\kappa_{\bmA\bmB} = \mp \frac{\sqrt{2}}{3}x_{\bmA\bmB} \mp
\frac{2\sqrt{2}m}{3r}x_{\bmA\bmB} + o_\infty(r^{-1/2}),
\label{AsymptoticFormKillingSpinor}
\end{equation}
with 
\[
x_{\bmA\bmB} \equiv \frac{1}{\sqrt{2}}
\left(
\begin{array}{cc}
-x^1 + \mbox{i} x^2 & x^3 \\
x^3 & x^1 +\mbox{i} x^2
\end{array}
\right).
\]
From the the above expressions one finds that on the asymptotic ends
\begin{eqnarray*}
&& \xi = \pm \sqrt{2} + o_\infty(r^{-1/2}), \\
&& \xi_{\bmA\bmB} = o_\infty(r^{-1/2}),
\end{eqnarray*}
where $\xi_{\bmA\bmB} \equiv
\epsilon_\bmA{}^A\epsilon_\bmB{}^B\xi_{AB}$. Moreover, for any electrovacuum initial data
set satisfying the conditions
\eqref{AsymptoticEuclidean1}-\eqref{AsymptoticEuclidean4} a spinor of
the form \eqref{AsymptoticFormKillingSpinor} satisfies
\[
\mathcal{D}_{(\bmA\bmB} \kappa_{\bmC\bmD)} = o_\infty(r^{-3/2}).
\]

\subsubsection{Asymptotic behaviour for non-Kerr data}

Not unsurprisingly, given electrovacuum initial data satisfying the conditions
\eqref{AsymptoticEuclidean1}-\eqref{AsymptoticEuclidean4}, it is always
possible to find a spinor $\kappa_{AB}$ satisfying the expansion
\eqref{AsymptoticFormKillingSpinor} in the asymptotic region. More precisely, one has:

\begin{lemma}
\label{Lemma:BoundaryConditions}
For any asymptotic end of an electrovacuum initial data set satisfying
\eqref{AsymptoticEuclidean1}-\eqref{AsymptoticEuclidean4} there exists
a spinor $\kappa_{AB}$ such that 
\[
\kappa_{AB} = \mp \frac{\sqrt{2}}{3}x_{\bmA\bmB} \mp
\frac{2\sqrt{2}m}{3r}x_{\bmA\bmB}+o_\infty(r^{-1/2})
\]
with
\begin{subequations}
\begin{eqnarray}
&& \xi = \pm \sqrt{2} + o_{\infty}(r^{-1/2}), \label{AnsatzBoundaryConditions1}\\
&& \xi_{AB} = o_\infty(r^{-1/2}), \label{AnsatzBoundaryConditions2}\\
&& \xi_{ABCD} = o_\infty(r^{-3/2}). \label{AnsatzBoundaryConditions3}
\end{eqnarray}
\end{subequations}
\end{lemma}

\begin{proof}
The proof follows the same structure of Theorem 17 in \cite{BaeVal10b}
where the vacuum case is considered. 
\end{proof}

\begin{remark}
\label{Remark:AnsatzSolution}
{\em The spinors obtained from the previous lemma can be cut-off so that
they are zero outside the asymptotic end. One can then add them to
yield a real spinor $\mathring{\kappa}_{AB}$ on the whole of
$\mathcal{S}$ such that 
\[
\mathcal{D}_{(AB} \mathring{\kappa}_{CD)} \in H^\infty_{-3/2}
\]
and asymptotic behaviour given by \eqref{AsymptoticFormKillingSpinor}
at each end.} 
\end{remark}

In the analysis of the solvability of the approximate Killing spinor
equation it is crucial that there exist no nontrivial
spatial Killing spinor that goes to zero at infinity. More precisely,
one has the following:

\begin{lemma}
\label{SKSEinjectivity}
Let $\nu_{AB}\in H^{\infty}_{-1/2}$ be a solution to
$\mathcal{D}_{(AB}\nu_{CD)}=0$ on an electrovacuum initial data set
satisfying the asymptotic conditions \eqref{AsymptoticEuclidean1}-\eqref{AsymptoticEuclidean4}. Then $\nu_{AB}=0$ on $\mathcal{S}$.
\end{lemma}
\begin{proof}
From Lemma \ref{Lemma:DerivativesKappa} one can write
$\mathcal{D}_{AB}\mathcal{D}_{CD}\mathcal{D}_{EF}\kappa_{GH} $ as a
linear combination of lower order derivatives, with smooth
coefficients. Direct inspection shows that the coefficients in this
linear combination have the decay conditions to make use of Theorem 20
from \cite{BaeVal10b} with $m=2$. It then follows that $\nu_{AB}$ must
vanish on $\mathcal{S}$.
\end{proof}

\subsection{Solving the approximate Killing spinor equation}

In the reminder of this section we will consider solutions to the
approximate Killing spinor equation of the form:
\begin{equation}
\kappa_{AB} = \mathring{\kappa}_{AB} + \theta_{AB}, \qquad
\theta_{AB}\in H^\infty_{-1/2}
\label{AnsatzApproximateKillingSpinor}
\end{equation}
with $\mathring{\kappa}_{AB}$ the spinor discussed in Remark
\ref{Remark:AnsatzSolution}.  For this Ansatz one has the following:

\begin{theorem}
\label{ExistenceSolutionsApproximateKillingSpinorEquation}
Given an electrovacuum asymptotically Euclidean initial data set
$(\mathcal{S},\bmh,\bmK,\bmE,\bmB)$ satisfying the asymptotic
conditions \eqref{AsymptoticEuclidean1}-\eqref{AsymptoticEuclidean4}
there exists a smooth unique solution to the approximate Killing
spinor equation \eqref{ApproximateKillingSpinorEquation} of the form
\eqref{AnsatzApproximateKillingSpinor}. 
\end{theorem}

\begin{proof}
The proof is analogous to that of Theorem 25 in \cite{BaeVal10b} and
is presented for completeness as this is the main result of this
article. 

Substitution of the Ansatz \eqref{AnsatzApproximateKillingSpinor} into
equation \eqref{ApproximateKillingSpinorEquation} yields the equation
\begin{equation}
L(\theta_{AB}) = - L(\mathring{\kappa}_{AB})
\label{EquationTheta}
\end{equation}
for the spinor $\theta_{AB}$. Due to elliptic regularity, any solution
to the above equation of class $H^2_{-1/2}$ is, in fact, a solution of class
$H^\infty_{-1/2}$. Thus, if a solution $\theta_{AB}$ exists then it
must be smooth. By construction ---see Remark
\ref{Remark:AnsatzSolution}--- it follows that $\mathcal{D}_{(AB}
\kappa_{CD)} \in H^\infty_{-3/2}$ so that 
\[
F_{AB} \equiv - L(\mathring{\kappa}_{AB}) \in H^\infty_{-5/2}.
\]   
In order to discuss the existence of solutions we make use of the
\emph{Fredholm alternative} for weighted Sobolev spaces. In the
particular case of equation \eqref{EquationTheta} there exists a
unique solution of class $H^{-1/2}$ if 
\[
\int_{\mathcal{S}} F_{AB} \widehat{\nu}^{AB} \mbox{d}\mu =0 
\]
for all $\nu_{AB}\in H^2_{-1/2}$ satisfying 
\[
L^*(\nu_{CD}) = L(\nu_{CD})=0.
\]
It will now be shown that a spinor $\nu_{AB}$ satisfying the above
must be trivial. Using the identity \eqref{IntegrationByParts} with 
$\zeta_{ABCD} = \mathcal{D}_{(AB} \nu_{CD)}$ and assuming that
$L(\nu_{CD})=0$ one obtains
\[
\int_{\mathcal{S}} \mathcal{D}^{AB} \nu^{CD}
\widehat{\mathcal{D}_{(AB} \nu_{CD)}}\mbox{d}\mu =
\int_{\partial\mathcal{S}_\infty} n^{AB} \nu^{CD} \widehat{\mathcal{D}_{(AB}\nu_{CD)}}\mbox{d}S
\]
where $\partial\mathcal{S}_\infty$ denotes the sphere at
infinity. Now, using that by assumption $\nu_{AB} \in H^2_{-1/2}$, it
follows that $\mathcal{D}_{(AB}\nu_{CD)} \in H^\infty_{-3/2}$ and that 
\[
n^{AB} \nu^{CD} \widehat{\mathcal{D}_{(AB}\nu_{CD)}}=o(r^{-2}). 
\]
The integration of the latter over a finite sphere is of type
$o(1)$. Accordingly, the integral over the sphere at infinity
$\partial\mathcal{S}_\infty$ vanishes and, moreover, 
\[
\int_{\mathcal{S}} \mathcal{D}^{AB} \nu^{CD}
\widehat{\mathcal{D}_{(AB} \nu_{CD)}}\mbox{d}\mu = 0.
\] 
Thus, one concludes that
\[
\mathcal{D}_{(AB} \nu_{CD)} =0 \qquad \mbox{over} \qquad \mathcal{S}
\]
so that $\nu_{AB}$ is a Killing spinor candidate. Now, Lemma \ref{SKSEinjectivity} shows that there are no non-trivial Killing
spinor candidates that go to zero at infinity. 

It follows from the discussion in the previous paragraph that the
kernel of the approximate Killing spinor operator is trivial and that
the Fredholm alternative imposes no obstruction to the existence of
solutions to \eqref{EquationTheta}. Thus, one obtains a unique
solution to the approximate Killing spinor equation with the
prescribed asymptotic behaviour at infinity.
\end{proof}

\section{The geometric invariant}
\label{Section:Invariant}

In this section we make use of the approximate Killing spinor
constructed in the previous section to construct
an invariant measuring the deviation of a given electrovacuum initial
data set satisfying the asymptotic conditions
\eqref{AsymptoticEuclidean1}-\eqref{AsymptoticEuclidean4} from initial
data for the Kerr-Newman spacetime.

\medskip
In the following let $\kappa_{AB}$ denote the approximate Killing
spinor obtained from Theorem
\ref{ExistenceSolutionsApproximateKillingSpinorEquation}, and let
\begin{subequations}
\begin{align}
J\equiv &\int_{\mathcal{S}} \mathcal{D}_{(AB} \kappa_{CD)}
  \widehat{\mathcal{D}^{AB}\kappa^{CD}}\mbox{d}\mu, \label{Invariant1}\\
I_{1}\equiv&\int_{\mathcal{S}}\Psi_{(ABC}{}^{F}\kappa_{D)F}\widehat{\Psi^{ABCG}\kappa^{D}{}_{G}} \mbox{d}\mu,  \label{Invariant2}\\
I_{2}\equiv&\int_{\mathcal{S}} \Theta_{AB}\widehat{\Theta^{AB}}\mbox{d}\mu,  \label{Invariant3}\\
I_{3}\equiv&\int_{\mathcal{S}}\mathcal{D}_{AB} \mathfrak{Q}^2
             \widehat{\mathcal{D}^{AB}\mathfrak{Q}^2}\mbox{d}\mu,  \label{Invariant4}
\end{align}
\end{subequations}
where following the notation of Section \ref{Section:KillingSpinorDataEquation} one has
\[
\Theta_{AB} \equiv 2 \kappa_{(A}{}^Q \phi_{B)Q}, \qquad \mathfrak{Q}^{2}\equiv \left(\kappa_{AB}\kappa^{AB} \right)^{2}
\phi_{AB}\phi^{AB}.
\]

The above integrals are well-defined. More precisely, one has that:

\begin{lemma}
Given the approximate Killing spinor $\kappa_{AB}$ obtained from
Theorem \ref{ExistenceSolutionsApproximateKillingSpinorEquation}, one
has that
\[
J,\, I_1,\, I_2,\, I_3 < \infty.
\]
\end{lemma}

\begin{proof}
By construction one has that the spinor $\kappa_{AB}$ obtained from
Theorem \ref{ExistenceSolutionsApproximateKillingSpinorEquation}
satisfies $\mathcal{D}_{(AB} \kappa_{CD)}\in H^0_{-3/2}$. It follows then
from the definition of the weighted Sobolev norm that
\[
\parallel \nabla_{(AB} \kappa_{CD)} \parallel_{H^0_{-3/2}} = \parallel
\nabla_{(AB} \kappa_{CD)} \parallel_{L^2} = J <\infty.
\]
To verify the the boundedness of $I_1$ one notices that by assumption
$\Psi_{ABCD} \in H^\infty_{-3+\varepsilon}$, $\kappa_{AB} \in
H^\infty_{1+\varepsilon}$ it follows by the multiplication properties
of weighted Sobolev spaces (see e.g. Lemma 14 in \cite{BaeVal10b})
that
\[
\Psi_{(ABC}{}^F\kappa_{D)F} \in H^\infty_{-3/2},
\]
so that, in fact, $I_1<\infty$.

We now look at the boundedness of $I_{2}$. By construction and due to the asymptotic conditions \eqref{AsymptoticEuclidean1}-\eqref{AsymptoticEuclidean4}, one can choose asymptotically Cartesian coordinates and orthonormal frames on the asymptotic ends such that the approximate Killing spinor and Maxwell spinor satisfy
\begin{align*}
\kappa_{\bmA\bmB} &= \mp\frac{\sqrt{2}}{3}x_{\bmA\bmB} + o_{\infty}\left(r^{1/2}\right) \\
\phi_{\bmA\bmB} &= \frac{q}{\sqrt{2}r^3}x_{\bmA\bmB} + o_{\infty}\left(r^{-5/2}\right)
\end{align*}
Therefore,
\begin{align*}
\Theta_{\bmA\bmB}&=\kappa_{(\bmA}{}^{\bmQ}\phi_{\bmB)\bmQ} \\
&= \mp\frac{q}{3r^3}x_{(\bmA}{}^{\bmQ}x_{\bmB)\bmQ} + o_{\infty}\left(r^{-3/2}\right) \\
&=o_{\infty}\left(r^{-3/2}\right)
\end{align*}
and so $\Theta_{AB}\in H_{-3/2}^{\infty}$, and $I_{2}<\infty$.

Finally, to show the boundedness of $I_{3}$, note that in the asymptotically Cartesian coordinates and orthonormal frames used above,we have
\begin{align*}
\left(\kappa_{AB}\kappa^{AB}\right)^2 &=\frac{4}{81}r^4 + o_{\infty}\left(r^{-7/2}\right) \\
\phi_{AB}\phi^{AB} &= \frac{q^2}{2r^4} + o_{\infty}\left(r^{-9/2}\right)
\end{align*}
and so the quantity $\mathfrak{Q}$ satisfies
\begin{equation*}
\mathfrak{Q}^{2}=\frac{2}{81}q^2 + o_{\infty}\left(r^{-1/2}\right)
\end{equation*}
Taking a derivative, one obtains
\begin{equation*}
\mathcal{D}_{AB}\mathfrak{Q}^2 = o_{\infty}\left(r^{-3/2}\right)
\end{equation*}
and therefore $\mathcal{D}_{AB}\mathfrak{Q}^2\in H_{-3/2}^{\infty}$ and $I_{3}<\infty$.
\end{proof}

\medskip
The integrals $J$, $I_1$, $I_2$ and $I_3$ are then used to define the
following geometric invariant:
\begin{equation}
I=J+I_{1}+I_{2}+I_{3}.
\label{GeometryInvariant}
\end{equation}

One has the following result combine the whole analysis of this
article:

\begin{theorem}
\label{MainTheorem:CharacterisationKerrNewmanData}
Let $(\mathcal{S},\bmh,\bmK,\bmE,\bmB)$ denote a smooth asymptotically
Euclidean initial data set for the Einstein-Maxwell equations
satisfying the on each of its two asymptotic ends the decay conditions
\eqref{AsymptoticEuclidean1}-\eqref{AsymptoticEuclidean4} with
non-vanishing mass and electromagnetic charge. Let $I$ be
the invariant defined by equation \eqref{GeometryInvariant} where
$\kappa_{AB}$ is the the unique solution to equation
\eqref{ApproximateKillingSpinorEquation} with asymptotic behaviour at
each end given by \eqref{AsymptoticFormKillingSpinor}. The invariant
$I$ vanishes if and only if $(\mathcal{S},\bmh,\bmK,\bmE,\bmB)$ is
locally an initial data set for the Kerr-Newman spacetime.
\end{theorem}

\begin{proof}
The proof follows the same strategy of Theorem 28 in \cite{BaeVal10b}. It
follows from our assumptions that if $I=0$ then the electrovacuum Killing spinor
data equations
\eqref{ReducedKillingSpinorDataCondition1}-\eqref{ReducedKillingSpinorDataCondition4}
are
satisfied on the whole of the hypersurface $\mathcal{S}$. Thus, from
Theorem \ref{Theorem:KillingSpinorData} the
development of the electrovacuum initial data
$(\mathcal{S},\bmh,\bmK,\bmE,\bmB)$ will have, at least on a slab a
Killing spinor. 

Now, the idea is to make use of Theorem
\ref{Theorem:KerrNewmanCharacterisation} to conclude that the
development will be the Kerr-Newman spacetime. For this, one has to conclude that
the spinor $\xi_{AA'}\equiv \nabla^Q{}_A \kappa_{BQ}$ is Hermitian so
that it corresponds to the spinorial counterpart of a real Killing
vector. By assumption, it follows from the expansions
\eqref{AnsatzBoundaryConditions1}-\eqref{AnsatzBoundaryConditions3} that
\[
\xi-\hat{\xi} = o_\infty (r^{-1/2}), \qquad \xi_{AB} + \hat{\xi}_{AB}
= o_\infty(r^{-1/2}).
\]
Together, the last two expressions correspond to the Killing initial
data for the imaginary part of $\xi_{AA'}$ ---thus, the imaginary part
of $\xi_{AA'}$ goes to zero at infinity. It is well know that for
electrovacuum spacetimes there exist no non-trivial Killing vectors of
this type \cite{BeiChr96,ChrOMu81}. Thus, $\xi_{AA'}$ is the spinorial
counterpart of a real Killing
vector.  By construction, $\xi_{AA'}$ tends, asymptotically, to a time
translation at infinity. Accordingly, the development of the
electrovacuum initial data   $(\mathcal{S},\bmh,\bmK,\bmE,\bmB)$
contains two asymptotically stationary flat ends $\mathcal{M}_\infty$
and $\mathcal{M}_\infty'$ generated by the Killing spinor
$\kappa_{AB}$. As the Komar mass and the electromagnetic charge of
each end is, by assumption, non-zero, one concludes from Theorem
\ref{Theorem:KerrNewmanCharacterisation} that the development
$(\mathcal{M},\bmg,\bmF)$ is locally isometric to the Kerr-Newman spacetime.
\end{proof}

\section{Conclusions}

As a natural extension to the vacuum case described by Backdahl and Valiente Kroon \cite{BaeVal10b}, the formalism presented above for the electrovacuum case has similar applications and possible modifications. For example, the use of asymptotically hyperboloidal rather than asymptotically flat slices can now be analysed for the full electovacuum case, applying to the more general Kerr-Newman solution. Another interesting alternative to asymptotically flat slices would be to obtain necessary and sufficient conditions for the existence of a Killing spinor in the future development of a pair of intersecting null hypersurfaces. For instance, one could take a pair of event horizons intersecting at a bifurcation surface, and obtain a system of conditions intrinsic to the horizon that ensures the black hole interior is isometric to the Kerr-Newman solution.

A motivation for the above analysis was also to provide a way of tracking the deviation of initial data from exact Kerr-Newman data in numerical simulations. However, in order to be a useful tool, one would still have to show that the geometric invariant is suitably behaved under time evolution (such as monotonicity). As highlighted in \cite{BaeVal10b}, a major problem is that it is hard to find a evolution equation for $\kappa_{AB}$ such that the elliptic equations \eqref{ApproximateKillingSpinorEquation} is satisfied on each leaf in the foliation. If these issues can be resolved, then this formalism may be of some use in the study of non-linear perturbations of the Kerr-Newman solution and the black hole stability problem.

Finally, the ethos of this article is to show that the characterisation of black hole spacetimes using Killing spinors is still a fruitful avenue of investigation. In the future, we hope to show that this method can be used to investigate other open questions, such as the Penrose inequality and black hole uniqueness.


\end{document}